\documentclass[onecolumn,draftcls,12pt]{IEEEtran}
\usepackage{verbatim}
\usepackage{amsfonts}
\usepackage{amssymb}
\usepackage{stfloats}
\usepackage{cite}
\usepackage{psfrag}
\usepackage{subfigure}
\usepackage{amsmath}
\usepackage{array}
\usepackage{epstopdf}
\usepackage{authblk}
\usepackage{graphicx} 
\usepackage{amsthm} 
\usepackage{lipsum}
\usepackage{verbatim} 
\usepackage{authblk}
\usepackage{mathtools}
\usepackage{cuted}
\usepackage[lined,boxed,ruled]{algorithm2e}
\usepackage{booktabs}
\usepackage{subfigure}
\usepackage{siunitx}
\usepackage{mathtools}
\usepackage{soul} 
\usepackage[]{mdframed}
\usepackage{setspace}

\newcommand{\bmu}{\boldsymbol{\mu}}

\newcommand{\bphi}{\boldsymbol{\phi}}
\newcommand{\balpha}{\boldsymbol{\alpha}}

\newtheorem{Lemma}{Lemma}

\newtheorem{Proposition}{Proposition}
\newtheorem{Remark}{Remark}

\newtheorem{Example}{Example}
\newtheorem{Corollary}{Corollary}
\newtheorem{MR}{Main Result}

\title{On the View-and-Channel Aggregation Gain in Integrated Sensing and Edge AI}
\author{}

\makeatletter
\newcommand{\removelatexerror}{\let\@latex@error\@gobble}
\makeatother

\begin{document}
 \author{{Xu~Chen},~{Khaled B. Letaief}, and~{Kaibin~Huang}

 \thanks{X. Chen and K. Huang are with the Department of Electrical and Electronic Engineering, The University of Hong Kong, Hong Kong. Khaled B. Letaief is with the Department of Electronic and Computer Engineering, Hong Kong University of Science and Technology, Hong Kong. Corresponding author: K. Huang (Email: huangkb@eee.hku.hk).}
}
\maketitle
\begin{abstract}
Sensing and edge \emph{artificial intelligence} (AI) are two key features of the \emph{sixth-generation} (6G) mobile networks. Their natural integration, termed \emph{Integrated sensing and edge AI} (ISEA), is envisioned to automate wide-ranging \emph{Internet-of-Tings} (IoT) applications. To achieve a high sensing accuracy, features of multiple sensor views are uploaded to an edge server for aggregation and inference using a large-scale AI model. The view aggregation is realized efficiently using over-the-air computing (AirComp), which also aggregates channels to suppress channel noise. As ISEA is at its nascent stage, there still lacks an analytical framework for quantifying the fundamental performance gains from view-and-channel aggregation, which motivates this work. Our framework is based on a well-established distribution model of multi-view sensing data where the classic Gaussian-mixture model is modified by adding sub-spaces matrices to represent individual sensor observation perspectives. Based on the model and linear classification, we study the End-to-End sensing (inference) uncertainty, a popular measure of inference accuracy, of the said ISEA system by a novel, tractable approach involving designing a scaling-tight uncertainty surrogate function, global discriminant gain, distribution of receive Signal-to-Noise Ratio (SNR), and channel induced discriminant loss. As a result, we prove that the E2E sensing uncertainty diminishes at an \emph{exponential} rate as the number of views/sensors grows, where the rate is proportional to global discriminant gain. Given AirComp and channel distortion, we further show that the exponential scaling remains but the rate is reduced by a linear factor representing the channel induced discriminant loss. Furthermore, in the case of many spatial degrees of freedom, we benchmark AirComp against equally fast, traditional analog orthogonal access. The comparative performance analysis reveals a sensing-accuracy crossing point between the schemes corresponding to equal receive array size and sensor number. 
This leads to the proposal of a scheme for adaptive access-mode switching to enhance ISEA performance. Last, the insights from our framework are validated by experiments using a convolutional neural network model and real-world dataset. 
\end{abstract}

\section{Introduction}
In June 2023, the International Telecommunication Union (ITU-R) finalized six major usage scenarios for 6G. While others represent a scaled-up version of 5G, two are new -- \emph{Integrated AI and Communications} (IAAC) and \emph{Integrated Sensing and Communication} (ISAC). IAAC reflects the 6G vision of edge intelligence, referring to ubiquitous distributed AI model training and inference at the network edge to support \emph{Internet-of-Tings} (IoT) applications~\cite{edgeAI_Letaief,WirelessmeetsML}. On the other hand, ISAC will leverage edge devices as distributed sensors and network-scale cooperation to enable 6G networks to have multi-view observations of the physical world in real-time~\cite{Sensing_in_6G, Letaief_CollaSensing2023}. The natural fusion of the two distinctive 6G functions, termed \emph{Integrated Sensing and Edge AI} (ISEA), shall provide a powerful platform for automating a broad range of IoT applications including auto-pilot, robotic control, digital twins, augmented reality, and localization and tracking~\cite{Letaief_CollaSensing2023,edgeAI_Letaief}. Unleashing the full potential of ISEA calls for a new goal-oriented design approach that integrates sensing, AI, and communication to optimize the end-to-end (E2E) performance~\cite{distributedAI,airpooling2023}. In this work, we contribute to the theoretic characterization of the E2E performance of ISEA, thereby laying a foundation for goal-oriented designs. 

A common backbone architecture for ISEA, called \emph{multi-view convolutional neural network} (MVCNN), wirelessly connects distributed sensors to an edge server~\cite{MVCNN}. Each sensor uses a lightweight neural network model for feature extraction from local sensing data and then uploads the local features for aggregation and inference at the server using a pre-trained deep neural network model supporting multi-modal computer vision~\cite{CV_for_locolization2022,Multiagent_perception,airpooling2023}. Local and server models are jointly trained as a single global model to maximize the E2E sensing (or inference) accuracy. This pertains to the common approach in edge AI, called split inference~\cite{DIB_JZhang}. By treating local and server models as components splitting the global model, relevant techniques can enable the adaptation of the splitting point to balance the device computation load and performance requirements in terms of, e.g., E2E latency and communication resources~\cite{Accelarate_edgeAI,SZhou_split_energy,Bennis_split_latency}. The mentioned feature aggregation, commonly referred to as multi-view pooling, is a key operation of MVCNN that exploits multiple sensor observations to improve sensing accuracy. The server operation fuses received local features into an aggregated feature map that is input into the global model (e.g., classifier) to generate a label identifying a target object/event. Element-wise averaging and maximization over local feature vectors are two popular aggregation functions termed average-pooling and max-pooling, respectively (see, e.g.,~\cite{airpooling2023}). Via view aggregation, multi-view sensing can attain an accuracy significantly higher than that of the single-view case especially when there are many sensors~\cite{MVCNN,MVsensing_graph}. However, the implementation of ISEA is confronted by a communication bottleneck resulting from the transmission of high-dimensional features by potentially a large cluster of sensors. 

Massive access techniques for 5G are insufficient for tackling the communication bottleneck of edge AI, which includes ISEA as a special case. Such techniques, for example, grant-free massive access, assume low-rate sporadic transmission by many low-complexity sensors monitoring environmental variables such as humidity and temperature~\cite{Larsson_grantfree,WeiYu_grantfree}. In contrast, 6G sensors are usually multi-modal devices (e.g., cameras and LIDAR) deployed in data-intensive computer vision applications such as surveillance, autonomous driving, and drone swarms~\cite{Chen2016Multiview}. The challenges are escalated by the tactile applications targeted by 6G, such as augmented reality and remote robotics, which demand air latency below 1 milli-second~\cite{3GPP_AItraffic}. The search for solutions motivates researchers to depart from the traditional communication-computing separation approach and advocate a paradigm shift towards the mentioned goal-oriented designs that target a specific task, such as distributed learning or sensing, and aim at maximizing the corresponding E2E system performance~\cite{edgeAI_Letaief, WirelessmeetsML}. One natural design approach for new paradigm is to customize existing techniques from multi-view sensing and edge AI, for example, sensor scheduling~\cite{When2com}, feature compression~\cite{DIB_JZhang}, and hierarchical pooling~\cite{feature_selection}, using an E2E metric (e.g., E2E sensing accuracy or latency) and targeting a specific air-interface technology (e.g., MIMO, OFDMA, and adaptive power control). An alternative, more revolutionary approach is to design new physical-layer technologies fully integrating computing and communication. In this vein, a representative class of techniques as considered in this work, called \emph{over-the-air computation} (AirComp), integrates multi-access and nomographic functional computation (e.g., averaging and geometric mean) to solve the scalability problem in traditional multi-access that divides radio resources~\cite{GXZhuAirComp2019}. AirComp's basic principle is to exploit the wave superposition property to achieve over-the-air aggregation of uncoded analog signals simultaneously transmitted by multiple devices. The scalability as a result of simultaneous access makes AirComp a popular air-interface technology for supporting fast and efficient distributed computing in 6G operations such as distributed learning~\cite{Air_distlearning}, inference~\cite{Air_edgeinference}, and sensing~\cite{airpooling2023,Task_oriented_DZ}. In addition, the use of uncoded analog transmission in AirComp is another factor contributing to the technology's ultra-low-latency while the resultant unreliability can be coped with by the robustness of data-analytics techniques or an AI algorithm~\cite{airpooling2023,chen2022analog,Marzetta_analogfb}. In particular, AirComp has been extensively studied for implementing over-the-air aggregation of local model updates in federated learning (FL) systems, leading to the emergence of an area called over-the-air FL~\cite{MZChen_JSACoverview}. Diversified design issues have been investigated including gradient sparsification~\cite{Air_distlearning}, beamforming~\cite{Ding2020TWC}, precoding~\cite{Eldar2021TSP}, power control~\cite{MXTao_powercontrol}, broadband transmission~\cite{GX_broadband}. 

Most recently, researchers also explored the applications of AirComp to realize over-the-air view aggregation in ISEA systems~\cite{airpooling2023,Task_oriented_DZ,ISACC_air2023,ISCC_DZ}. In~\cite{airpooling2023}, max-pooling, which is not directly AirComputable, is realized using AirComp via p-norm approximation of maximization. The parameter of the approximation function is optimized to balance the noise effect and approximation error. The optimization still uses the generic metric of AirComp error (i.e., the error in computed function values with respect to the noiseless case) instead of the E2E sensing accuracy though the two metrics are related by a derived inequality. Similarly, the AirComp error is adopted in~\cite{ISACC_air2023} as the performance metric to optimize a receive beamformer in a system supporting integrated MIMO radar sensing and AirComp. On the other hand, a different metric, discrimination gain, has been proposed to approximately measure the sensing accuracy with tractability. In~\cite{ISCC_DZ}, the effects of sensing, computation, and communication on the discrimination gain are quantified and controlled by designing a task-oriented resource management approach so as to optimize the E2E performance. The average discrimination gain for an individual feature dimension is further considered in~\cite{Task_oriented_DZ} to facilitate importance-aware beamforming that adapts effective channel gains of different sensors according to their importance levels accounting for both average discriminant gains and channel states. Wireless for ISEA is still a nascent area where prior work largely focuses on algorithmic designs. There still lacks a systematic framework for analyzing E2E performance. Specifically, in the aspect of multi-view sensing, the sensing accuracy sees continuous improvements with the growth of the number of sensors providing view diversity. There exist few results on quantifying the scaling law. On the other hand, in the aspect of air interface, the AirComp error diminishes with the increase of the number of links due to aggregation and exploitation of (channel) spatial diversity~\cite{Air_diversity}. The consideration of E2E performance for ISEA naturally couples the two aspects and gives rise to the following open research questions we attempt to answer sequentially in this work.
\begin{enumerate}
    \item \textbf{(View Aggregation)} Consider ISEA without channel distortion. \emph{How does the accuracy of multi-view sensing improve as the number of sensors grows?}
    \item \textbf{(View-and-Channel Aggregation)} Consider ISEA with wireless channels and using AirComp. \emph{How does the E2E sensing accuracy improve as the number of sensors grows?}
    \item \textbf{(Optimality of AirComp)} AirComp supports simultaneous access when spatial \emph{degrees of freedom} (DoFs) are insufficient for orthogonal access.  When many spatial DoFs are available, is AirComp still optimal for fast ISEA? 
\end{enumerate}

By making an attempt to answer these questions, we derive a theoretic framework for quantifying the E2E performance of an ISEA system implemented on the MVCNN architecture with an AirComp-based air interface. Key models and assumptions are summarized as follows. First, a well-established mathematical model for multi-view sensing is adopted~\cite{Multiview_analysis_2018, Multiview_analysis_2019}. In this real-data validated model, features extracted from sensor observations (e.g., images) are described as low-rank projections of a high-dimensional ground-truth feature map, where a projection matrix, called \emph{observation matrix}, reflects the spatial relationship between the associated sensor and the target object. Second, the feature map is assumed to distribute following the classic Gaussian mixture model (GMM) widely used in statistical learning~\cite{statistic_learning_book} and deep learning (see, e.g., ~\cite{GMM_pose_esti}). The model comprises multiple Gaussian clusters, each of which is tagged with an object-class label. Third, channel coefficients of the multiuser \emph{single-input-multi-output} (SIMO) uplink channel are assumed to be independent and identically distributed (i.i.d.) Rayleigh fading, representing spatial diversity from rich scattering and spatially separated sensors. Last, the E2E sensing accuracy is measured by the popular metric of sensing uncertainty that is computed as the entropy of posteriors of object classes conditioned on observations~\cite{ProgressFTX_Qiao,Entropy2ErrorPorb}.

Then the key contributions and findings of this work are summarized as follows.
\begin{itemize}
    \item \textbf{View Aggregation Gain}: To answer Research Question 1 for the noise-free case, the E2E sensing uncertainty is derived as a function of multiple factors including the number of sensors, number of (object) classes, and average differentiability of class pairs measured using the well known Mahalanobis distance. The average class differentiability is \emph{with respect to} (w.r.t.) the feature sub-space defined using the \emph{global observation matrix} that cascades the observation matrices of all sensors. The derivation exploits the tractability of GMM to derive asymptotically tight bounds on sensing uncertainty. The derived function reveals that view-aggregation gains in two aspects. On one hand, the monotonic reduction of sensing uncertainty w.r.t. the sensor population reflects its resultant suppression of \emph{sensing noise}. On the other hand, the function is also a monotone decreasing w.r.t. the average class differentiability that is in turn enhanced as a growing number of suitably scheduled sensors boosts the rank of the global observation matrix. When the number of views is large, the uncertainty function is shown to exhibit a simplified form \emph{linearly proportional} to the number of classes but diminish at an \emph{exponential rate} linearly proportional to the number of views/sensors. 

    \item \textbf{View-and-Channel Aggregation Gain}: To address Research Question 2, we consider ISEA with channel distortion induced by AirComp for implementing multi-view aggregation. Building on the preceding analysis and applying random-matrix theory, the sensing uncertainty is shown to scale similarly as its noiseless counterpart except for an additional linear scaling factor for the exponential decay rate. The factor represents the negative channel effect on average class differentiability and is proved to be a monotone decreasing function of the effective receive SNR after AirComp. The analysis reveals that as the number of sensors increases, aggregation suppresses noise sufficiently fast such that the exponential decay of sensing uncertainty in the noiseless case is retained albeit at a slower exponential rate. 

    \item \textbf{AirComp versus Orthogonal Access}: To answer Research Question 3, AirComp is benchmarked against analog orthogonal access, both of which support low-latency uncoded analog transmission~\cite{chen2022analog,Marzetta_analogfb}. AirComp's main advantage lies in supporting spatial simultaneous access even when spatial DoFs are insufficient for orthogonal access. However, as the receive array size increases, we show the existence of a crossing point (with array size and number of sensors approximately equal) above which analog orthogonal access outperforms AirComp in terms of sensing uncertainty. This motivates an adaptive scheme that switches between AirComp and analog orthogonal access depending on the available spatial DoFs.

    \item \textbf{Experiments}: The preceding analytical results are validated in ISEA experiments using both synthetic (i.e., GMM) and real datasets (i.e., ModelNet~\cite{MVCNN}).
\end{itemize}

The remainder of this paper is organized as follows. The multi-view sensing and communication models are elaborated in Section~\ref{section:model}. The analysis results for noiseless and AirComp-based view aggregation cases are presented in Section~\ref{section:analysis_1} and Section~\ref{section:analysis_2}, respectively. In Section~\ref{section:comparison}, we benchmark AirComp against analog orthogonal access, followed
by performing experiments in Section~\ref{section:experiments}.

%
\begin{figure*}[t]
    \centering
    \includegraphics[width=0.96\textwidth]{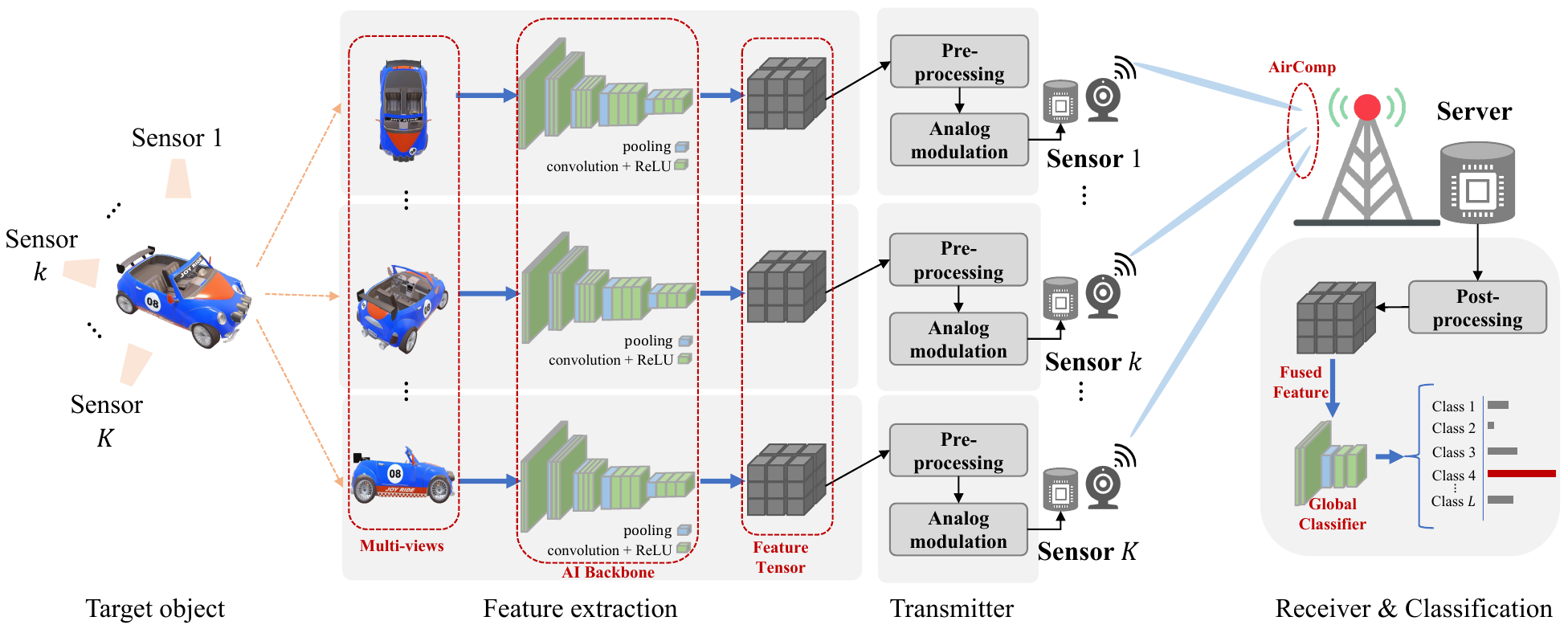}
    \caption{A system integrating multi-view sensing and edge AI.}
    \label{fig:system_model}
\end{figure*}
\section{System, Models, and Metrics}\label{section:model}
Consider an ISEA system where a server realizes remote object detection by leveraging AI-based multi-view sensing over $K$ distributed sensors, as illustrated in Fig.~\ref{fig:system_model}.  
Relevant operations, models, and metrics are described in the following sub-sections.
\subsection{Multi-View Sensing Model}
\subsubsection{Local Data Distribution} Each sensor, say sensor $k$, feeds its captured raw data (e.g. images) into a pre-trained model to generate a feature map, denoted by $\mathbf{f}_k\in\mathbb{R}^M$, that comprises $M$ real features. The distribution of $\mathbf{f}_k$ is given as follows. First, let $\mathbf{g}\in\mathbb{R}^M$ be the ground-truth feature map corresponding to the current object and be assumed to have a uniform prior distribution over $L$ classes~\cite{Task_oriented_DZ,ProgressFTX_Qiao}:
\begin{equation}\label{eq:uniform_prior}
    \mathrm{Pr}\left(\mathbf{g}=\bmu_{\ell}\right) = \frac{1}{L},\ \forall \ell,
\end{equation}
where $\bmu_{\ell}$ denotes the centroid of the $\ell$-th class in the feature space. Then, due to the limited physical view of the sensors, $\mathbf{f}_k$ represents a low-dimensional projection of $\mathbf{g}$~\cite{Multiview_analysis_2018, Multiview_analysis_2019}. Adopting a well-established multi-view sensing model in the literature (see, e.g.~\cite{Multiview_analysis_2018}), we can relate the feature map $\mathbf{f}_k$ to $\mathbf{g}$ as
\begin{equation}\label{eq:local_feature}
    \mathbf{f}_k = \mathbf{P}_{k}\mathbf{g} + \mathbf{w}_k,
\end{equation}
where $\mathbf{P}_{k}$ is the low-rank observation matrix of sensor $k$ and $\mathbf{w}_k$ represents the inherent sensing noise following an \emph{independent and identically distributed} (i.i.d.) Gaussian distribution $\mathcal{N}(\mathbf{0},\mathbf{C})$. The observation matrices $\{\mathbf{P}_k\}$ and the covariance matrix $\mathbf{C}$ can be learned by using subspace-representation networks~\cite{Multiview_analysis_2018, Multiview_analysis_2019} and are considered to be available to both devices and the server. It follows from~\eqref{eq:local_feature} that local feature maps follow the distribution of a \emph{Gaussian mixture model} (GMM)~\cite{statistic_learning_book}:
\begin{equation}\label{eq:local_PDF}
    \mathbf{f}_k\sim\frac{1}{L}\sum_{\ell=1}^{L}\mathcal{N}(\mathbf{P}_{k}\bmu_{\ell},\mathbf{C}),
\end{equation}
where $\mathcal{N}(\mathbf{P}_{k}\bmu_{\ell},\mathbf{C})$ represents a Gaussian distribution with mean $\mathbf{P}_{k}\bmu_{\ell}$ and covariance $\mathbf{C}$.  

\subsubsection{Global Classification} Next, $\{\mathbf{f}_k\}$ are uploaded to the server for classification as follows.
They are first fused into a single feature map, denoted by $\bar{\mathbf{f}}$, which is known as view aggregation (or pooling). The popular average aggregation is adopted as $\bar{\mathbf{f}}=\frac{1}{K}\sum_{k=1}^K\mathbf{f}_k$~\cite{MVCNN}. Then, $\bar{\mathbf{f}}$ is fed into a classifier for inference. We consider two types of classifiers.
\begin{itemize}
    \item \textbf{Linear Classification:} A linear classifier is considered in analysis for tractability~\cite{statistic_learning_book}. Consider a case of two classes ($L=2$), the linear classifier distinguishes the pair of classes by using a classification boundary between their clusters, which is defined as a hyperplane in the feature space
    \begin{equation}
        \mathcal{H}(\balpha,\beta) = \{\mathbf{f}:\balpha^{\top}\mathbf{f} + \beta=0\}.
    \end{equation}
    The optimal label of an input feature map $\bar{\mathbf{f}}$ is assigned as one class if $\bar{\mathbf{f}}$ is determined to be above the hyperplane (i.e., $\balpha^{\top}\bar{\mathbf{f}} + \beta\geq 0$); otherwise, $\bar{\mathbf{f}}$ is labeled as the other class.
    In a general case with $L>2$, there are $L(L-1)/2$ classification boundaries, and the optimal result can be obtained via sequential conduction of the one-versus-one classification. Given equal priors of classes, $\{\mathbf{f}_k\}$ follow the distribution in~\eqref{eq:local_PDF}. The optimal $L$-class linear classifier is the \emph{maximum likelihood} (ML) design~\cite{statistic_learning_book}: 
    \begin{equation}\label{eq:ML_classifier}
        \ell^{\star}=\arg\max_{\ell}\ \log\mathrm{Pr}\left(\bar{\mathbf{f}}|\bmu_{\ell}\right).
    \end{equation}
    
    \item \textbf{CNN Classification:} MVCNN model is adopted in experiments. The model consists of two parts, $\mathcal{F}_1$ and $\mathcal{F}_2$, that are employed at sensors and the server, respectively. The sub-model  $\mathcal{F}_1$ is identical for sensors and used to extract local features from sensing data. After view aggregation, the obtained aggregated feature vector $\bar{\mathbf{f}}$ is fed into $\mathcal{F}_2$ that outputs scores for individual classes. The class with the highest score is selected as the prediction result. 
\end{itemize}  

\subsection{Multi-Access Models}
For the ISEA system in Fig. 1, we mainly consider analog multi-access techniques for enabling efficient simultaneous access (i.e., view-and-channel aggregation). We also explore ISEA with noiseless feature aggregation in Section III, aligning with scenarios involving reliable digital transmission. In the class of analog transmission, we primarily adopt AirComp for feature aggregation. To investigate its optimality, we further consider analog orthogonal access, namely orthogonal access with fast analog transmission~\cite{Marzetta_analogfb}, as a benchmark scheme. It achieves the same multi-access latency as AirComp but requires receive spatial DoF to be equal to or exceed the number of sensors. The assumptions and operations of the schemes are described as follows. The server and sensors are equipped with $N$-element array and a single antenna, respectively. Assuming a frequency non-selective channel, time is slotted and each slot is used for transmitting one symbol. Block fading is considered such that the channel remains unchanged over a coherence duration comprising $T$ time slots. Symbol-level synchronization is assumed over all sensors.

\subsubsection{Analog Transmission}
In an arbitrary time slot, say slot $t$, sensors simultaneously transmit their linear analog modulated data symbol, $\{x_{k,t}\}$, leading to the server receiving a symbol vector:
\begin{equation}\label{eq:model}
\mathbf{y}_t=\sum_{k}\rho_k\mathbf{h}_kx_{k,t}+\mathbf{z}_t,
\end{equation}
where $\rho_k$ represents transmit power, $\mathbf{h}_k\in\mathbb{C}^{N\times 1}$ denotes the channel vector of sensor $k$, and $\mathbf{z}_t\sim\mathcal{CN}(\mathbf{0},\sigma^2\mathbf{I}_N)$ models additive channel noise. 
Assuming Rayleigh fading, $\mathbf{h}_k$ is composed of i.i.d. $\mathcal{CN}(0,1)$ entries and is independent between sensors. Let $\nu^2\overset{\triangle}{=}\mathsf{E}\left[\frac{1}{T}\sum_{t=1}^Tx_{k,t}^2\right]$ be the variance of transmitted symbols over a channel coherence block. Each sensor is constrained by a power budget of $P$, i.e., $\rho_k^2\nu^2\leq P$. Then, the transmit SNR is defined as $\gamma = \frac{P}{\sigma^2}$. 

Let $\mathbf{x}_k = [x_{k,1},\cdots,x_{k,M}]$ denote the symbol vector transmitted from sensor $k$ over $M$ time slots with $M\leq T$. The data vector in~\eqref{eq:model} received over $M$ time slots can be aggregated into a matrix symbol, $\mathbf{Y} = [\mathbf{y}_1,\cdots,\mathbf{y}_M]$:
\begin{equation}\label{eq:symbol_matrix}
    \mathbf{Y} = \sum_{k}\rho_k\mathbf{h}_k\mathbf{x}_k+\mathbf{Z},
\end{equation}
where $\mathbf{Z} = [\mathbf{z}_1,\cdots,\mathbf{z}_M]$. 

\subsubsection{Receiver for AirComp}
The symbol vector $\mathbf{x}_k$ in~\eqref{eq:symbol_matrix} is computed as $\mathbf{x}_k=(\mathbf{f}_k-\mathbf{f}_{k}^{\mathsf{avg}})^{\top}$ with $\mathbf{f}_{k}^{\mathsf{avg}} = \mathsf{E}\left[\mathbf{f}_k\right]$ to have zero mean. The parameter $\mathbf{f}_{k}^{\mathsf{avg}}$ is also available for both the server to receive the features. Following the AirComp literature, \emph{zero-forcing} (ZF) transmit power control is adopted to overcome channel distortion, $\rho_k = (\mathbf{b}^H\mathbf{h}_k)^{-1}$ with $\mathbf{b}\in \mathbb{C}^{N}$ being a receive combiner (see, e.g.,~\cite{GXZhuAirComp2019}). The output is given as $\mathbf{s}=\mathbf{Y}^H\mathbf{b}=\sum_{k}\mathbf{f}_k-\mathbf{f}_{k}^{\mathsf{avg}} + \mathbf{Z}^H\mathbf{b}$. From $\mathbf{s}$, a noisy version of $\bar{\mathbf{f}}$, denoted by $\tilde{\mathbf{f}}$, can be obtained by the following post-processing:
\begin{equation}\label{eq:post_process}
    \begin{aligned}
        \tilde{\mathbf{f}}
        = \frac{1}{K}\mathbf{s} + \frac{1}{K}\sum_{k}\mathbf{f}_{k}^{\mathsf{avg}}
        = \bar{\mathbf{f}} + \frac{1}{K}\mathbf{Z}^H\mathbf{b}.
    \end{aligned}
\end{equation}
\subsubsection{Receiver for Orthogonal Multi-Access}
The simultaneous data streams in~\eqref{eq:symbol_matrix} are orthogonalized via receive beamforming. It is optimal to maximize transmit power at each sensor as $\rho_k = \frac{\sqrt{P}}{\nu}$. Then the received symbol matrix can be rewritten as $\mathbf{Y}=\sum_k\frac{\sqrt{P}}{\nu}\mathbf{h}_k\mathbf{x}_k + \mathbf{Z}$. Let $\mathbf{e}_k=[0,\cdots,1,\cdots,0]^{\top}$ denote the standard basis vector with the $k$-th element being 1. The data stream of sensor $k$, denoted by $\mathbf{s}_k$, can be extracted from  $\mathbf{Y}$ using a ZF beamformer $\mathbf{b}_k = \mathbf{H}(\mathbf{H}^H\mathbf{H})^{-1}\mathbf{e}_k$ with $\mathbf{H} = [\mathbf{h}_1,\cdots,\mathbf{h}_K]$~\cite{Taesang2006}:
\begin{equation}
\begin{aligned}
        \mathbf{s}_k= \mathbf{Y}^H\mathbf{b}_k= \frac{\sqrt{P}}{\nu}\left(\bar{\mathbf{f}}_k-\mathbf{f}_{k}^{\mathsf{avg}}\right)+\mathbf{Z}^H\mathbf{b}_k.
\end{aligned}
\end{equation}
By slight abuse of notation, let $\tilde{\mathbf{f}}$ also denote the noisy version of the desired aggregated feature vector in the current case. Then it can be obtained by the following post-processing:
\begin{equation}\label{eq:OA_analog}
    \begin{aligned}
    \tilde{\mathbf{f}}=\frac{1}{K}\sum_{k}\left(\frac{\nu}{\sqrt{P}}\mathbf{s}_k+\mathbf{f}_{k}^{\mathsf{avg}}\right)= \bar{\mathbf{f}} + \frac{\nu}{K\sqrt{P}}\mathbf{Z}^H\sum_k\mathbf{b}_k.
\end{aligned}
\end{equation}

\subsection{E2E Performance Metric}
Typically, the sensing (inference) accuracy is  defined as the probability of correct classification. For the purpose of tractable analysis, we also consider the following two relevant metrics for evaluating the E2E sensing performance.
\subsubsection{Sensing Uncertainty}
As a popular measure related to sensing (inference) accuracy, the metric is defined as the entropy of posteriors of classification classes given the aggregated feature~\cite{Entropy2ErrorPorb}. Mathematically, given the aggregated feature map $\tilde{\mathbf{f}}$, its sensing uncertainty, denoted by $H$, is given as
\begin{equation}\label{eq:uncertainty}
    H= \mathsf{E}_{\tilde{\mathbf{f}}}\left[-\sum_{\ell=1}^{L}\mathrm{Pr}\left(\bmu_{\ell}|\tilde{\mathbf{f}}\right)\log\mathrm{Pr}\left(\bmu_{\ell}|\tilde{\mathbf{f}}\right)\right].
\end{equation}

\subsubsection{Discrimination Gain}
The sensing accuracy is largely determined by the discernibility between a pair of classes that can be measured by \emph{discrimination gains} computed as their \emph{symmetric Kullback-Leibler (KL) divergence}~\cite{ProgressFTX_Qiao}. Considering sensor $k$, the local discrimination gain between classes $\ell$ and $\ell^{\prime}$, denoted as $G_k(\ell,\ell^{\prime})$, can be computed as
\begin{align}\label{eq:discrimination_gain}
G_k(\ell,\ell^{\prime})=&\mathsf{KL}\left(\mathcal{N}\left(\mathbf{P}_k\bmu_{\ell},\mathbf{C}\right)||\mathcal{N}\left(\mathbf{P}_k\bmu_{\ell^{\prime}},\mathbf{C}\right)\right) \nonumber\\
    &+ \mathsf{KL}\left(\mathcal{N}\left(\mathbf{P}_k\bmu_{\ell^{\prime}},\mathbf{C}\right)||\mathcal{N}\left(\mathbf{P}_k\bmu_{\ell},\mathbf{C}\right)\right)\nonumber\\
        =&\left(\bmu_{\ell}-\bmu_{\ell^{\prime}}\right)^{\top}\mathbf{P}_k\mathbf{C}^{-1}\mathbf{P}_k\left(\bmu_{\ell}-\bmu_{\ell^{\prime}}\right).
\end{align}
The global discriminant gain is derived in the next section.

%

\section{Multi-view Aggregation Gain
without Channel Distortion}\label{section:analysis_1}
In this section, we consider the scenario with the absence of channel noise and focus on analyzing the E2E performance of an ISEA system in terms of sensing uncertainty. The tractable analysis consists of three steps presented in separate sub-sections, namely characterizing the distribution of aggregated features, designing a suitable  surrogate function for sensing uncertainty, and deriving the scaling laws of sensing uncertainty.

\subsection{Aggregated Feature Distribution}
The computation of sensing uncertainty in~\eqref{eq:uncertainty} relies on an explicit distribution of the aggregated feature map, $\bar{\mathbf{f}}$,  at the input to the classifier. Based on the GMM model of local features, the desired result is derived  as shown below.

\begin{Lemma}[Distribution of Aggregated Feature Map]\label{Lemma:pooled_feature}
    \emph{Based on the distribution of local feature maps in~\eqref{eq:local_PDF} and in the absence of channel noise, the aggregated feature $\bar{\mathbf{f}}$ follows a Gaussian-mixture distribution given as
    \begin{equation*}
    \bar{\mathbf{f}}\sim\frac{1}{L}\sum_{\ell=1}^{L}\mathcal{N}\left(\bar{\mathbf{P}}\bmu_{\ell},\frac{1}{K}\mathbf{C}\right),
    \end{equation*}
    where $\bar{\mathbf{P}} = \frac{1}{K}\sum_k\mathbf{P}_{k}$ denotes the average  of local observation matrices, termed \emph{global observation matrix}.}
\end{Lemma}
\begin{proof}
    (See Appendix~\ref{Apdx:pooled_feature}).
\end{proof}

Comparing Lemma~\ref{Lemma:pooled_feature} and~\eqref{eq:local_PDF}, it is observed that the multi-view aggregation retains the distribution of features except for 1) replacing the cluster centroid of individual classes with their projections onto the global observation matrix and 2) narrowing the cluster size by reducing covariance by the factor $1/K$. These two factors contribute to the multi-view aggregation gain in sensing performance, which is quantified in the sequel. 

Using Lemma~\ref{Lemma:pooled_feature} and~\eqref{eq:ML_classifier}, the optimal linear classifier with  $\bar{\mathbf{f}}$ as input can be written as 
\begin{equation}\label{eq:noiseless_classifier}
\begin{aligned}
    \ell^{\star}=\arg\min_{\ell}\ \left(\bar{\mathbf{f}}-\bar{\mathbf{P}}\bmu_{\ell}\right)^{\top}\mathbf{C}^{-1}\left(\bar{\mathbf{f}}-\bar{\mathbf{P}}\bmu_{\ell}\right).
\end{aligned}
\end{equation}
The  uniqueness of the inferred label $\ell^{\star}$ is guaranteed by the independence among views and $\bar{\mathbf{f}}$ being a continuous random variable. 
\subsection{Surrogate Function for Sensing Uncertainty}
For the purpose of tractability, we propose a simpler but scaling-tight surrogate function for sensing uncertainty. To this end, let  $G_{\ell,\ell^{\prime}}$ denote the global discrimination gain of $\bar{\mathbf{f}}$. Using  the local discriminant gain in~\eqref{eq:discrimination_gain}, the expression of $G_{\ell,\ell^{\prime}}$ can be obtained as shown below. 

\begin{Lemma}[Global Discrimination Gain]\label{Lemma:global_gain}
    \emph{Using $\bar{\mathbf{f}}$ for classification, the global discrimination gain between class $\ell$ and $\ell^{\prime}$ is given as $G_{\ell,\ell^{\prime}} =  KD_{\ell,\ell^{\prime}}$ where 
    \begin{equation*}
    D_{\ell,\ell^{\prime}}=\left(\bmu_{\ell}-\bmu_{\ell^{\prime}}\right)^{\top}\bar{\mathbf{P}}\mathbf{C}^{-1}\bar{\mathbf{P}}\left(\bmu_{\ell}-\bmu_{\ell^{\prime}}\right).
    \end{equation*}}
\end{Lemma}

In Lemma~\ref{Lemma:global_gain}, $D_{\ell,\ell^{\prime}}$ is  equal to the Mahalanobis distance between classes $\ell$ and $\ell^{\prime}$, which reflects their   differentiability~\cite{Mahalanobis1936}. Then, leveraging the distribution of $\bar{\mathbf{f}}$ in Lemma~\ref{Lemma:pooled_feature} and optimization theory, the resulting sensing uncertainty is obtained as follows.
\begin{Proposition}[Sensing Uncertainty]\label{Prop:uncertainty}
    \emph{In the case of  linear classification, the E2E sensing uncertainty, $H$,  in~\eqref{eq:uncertainty}  can be bounded as
    \begin{equation*}
    \begin{aligned}
        &\frac{1}{L}\sum_{\ell=1}^{L}\log\left[1+\sum_{\ell^{\prime}\neq\ell}\exp\left(-\frac{D_{\ell,\ell^{\prime}}}{2}K\right)\right]\leq H\leq\\
        &
        \frac{1}{L}\sum_{\ell=1}^{L}\log\left[1+\sum_{\ell^{\prime}\neq\ell}\exp\left(-\frac{D_{\ell,\ell^{\prime}}}{cM+2}K\right)\right] +  C_a,
    \end{aligned}
    \end{equation*}
    where $c>0$ is arbitrary, the constant $C_a = \log \frac{c e^{\frac{1}{c}}}{1+c}$, and $D_{\ell,\ell^{\prime}}$ is given in Lemma~\ref{Lemma:global_gain}.}
\end{Proposition}
\begin{proof}
    (See Appendix~\ref{Apdx:uncertainty}).
\end{proof}

The relatively complex expression of $H$ in~\eqref{eq:uncertainty} does not allow tractable analysis of its scaling laws. Nevertheless, its bounds in  Proposition~\ref{Prop:uncertainty} suggests that $H$  can be approximated by the following  surrogate function: 
\begin{equation}\label{eq:surrogate}
    H_{\mathsf{s}} = \frac{1}{L}\sum_{\ell=1}^{L}\log\left[1+\sum_{\ell^{\prime}\neq\ell}\exp\left(-\kappa D_{\ell,\ell^{\prime}}K\right)\right],
\end{equation}
where $\kappa$ can be $\frac{1}{2}$ and $\frac{1}{cM+2}$ corresponding to the lower and upper bounds in Proposition~\ref{Prop:uncertainty}, respectively. The function is found to follow a similar scaling law as $H$, which is essential for subsequent asymptotic analysis. The scaling-tight property of $H_{\mathsf{s}}$ is validated  using the following numerical example.
\begin{Example}[Numerical Validation]\label{Example:toy}
\emph{Let the feature dimension and the number of classes be equal: $M=L=5$. The covariance of the local feature distribution is set as $\mathbf{C} = 0.1\mathbf{I}_M$. The curves  of exact sensing uncertainty and surrogate function in~\eqref{eq:surrogate} are plotted in Fig.~\ref{fig:approximation} for a variable number of views/sensors. By comparing the curves, one can observe that the  sensing uncertainty is not necessarily a monotonically decreasing function w.r.t. $K$ (e.g., from $K =2$ to $K =3$) when $K$ and the local observation DoFs, referring to $\mathsf{rank}(\mathbf{P}_k)$, are small. 
The monotonicity of multi-view aggregation w.r.t. $K$ emerges when local sensors acquire sufficiently many observation DoFs (i.e., $2$), as shown in Fig.~\ref{subfig:dim2}. 
Based on observations from Fig.~\ref{fig:approximation}, we can conclude that the surrogate function accurately captures the scaling law of sensing uncertainty including reflecting the said glitches of monotonicity. 
\begin{figure}[t]
    \centering
    \subfigure[Observation DoF = 1]{\label{subfig:dim1}\includegraphics[width=0.44\textwidth]{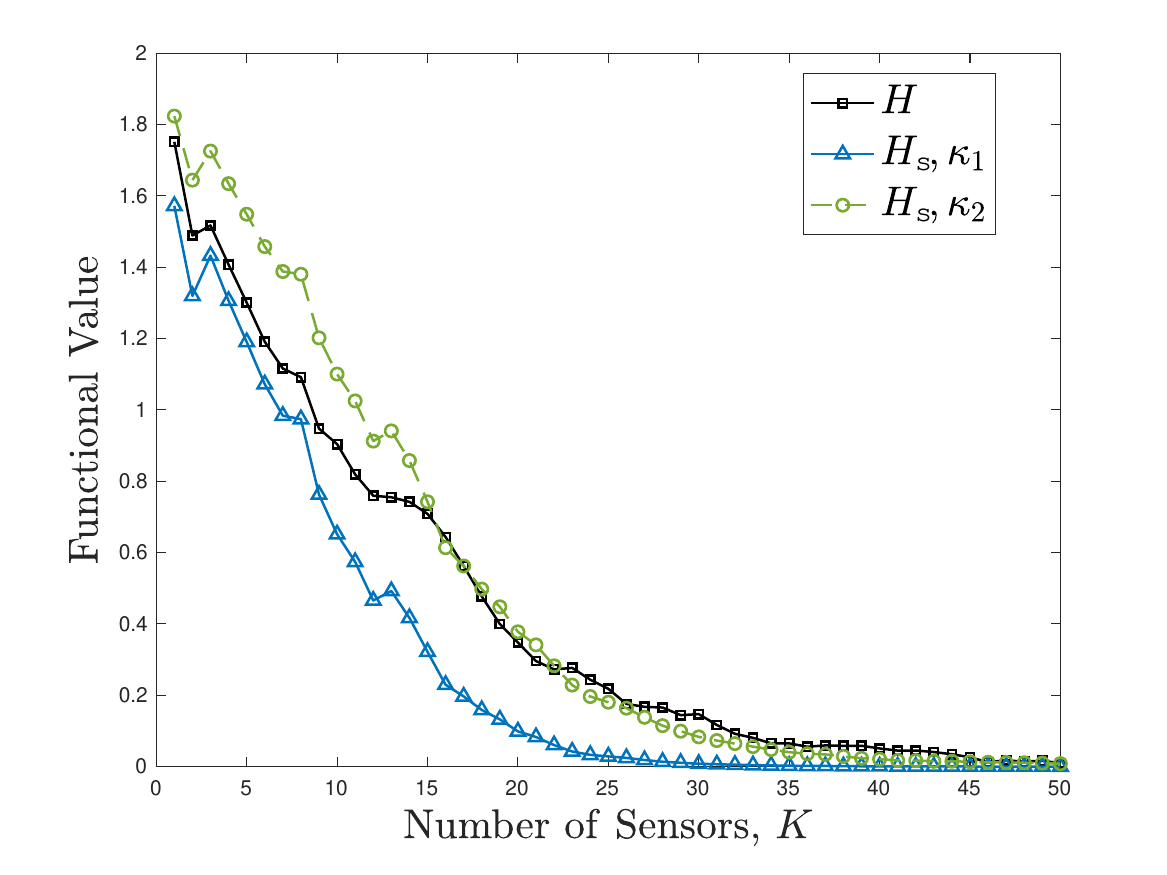}}
    \subfigure[Observation DoF = 2]{\label{subfig:dim2}\includegraphics[width=0.44\textwidth]{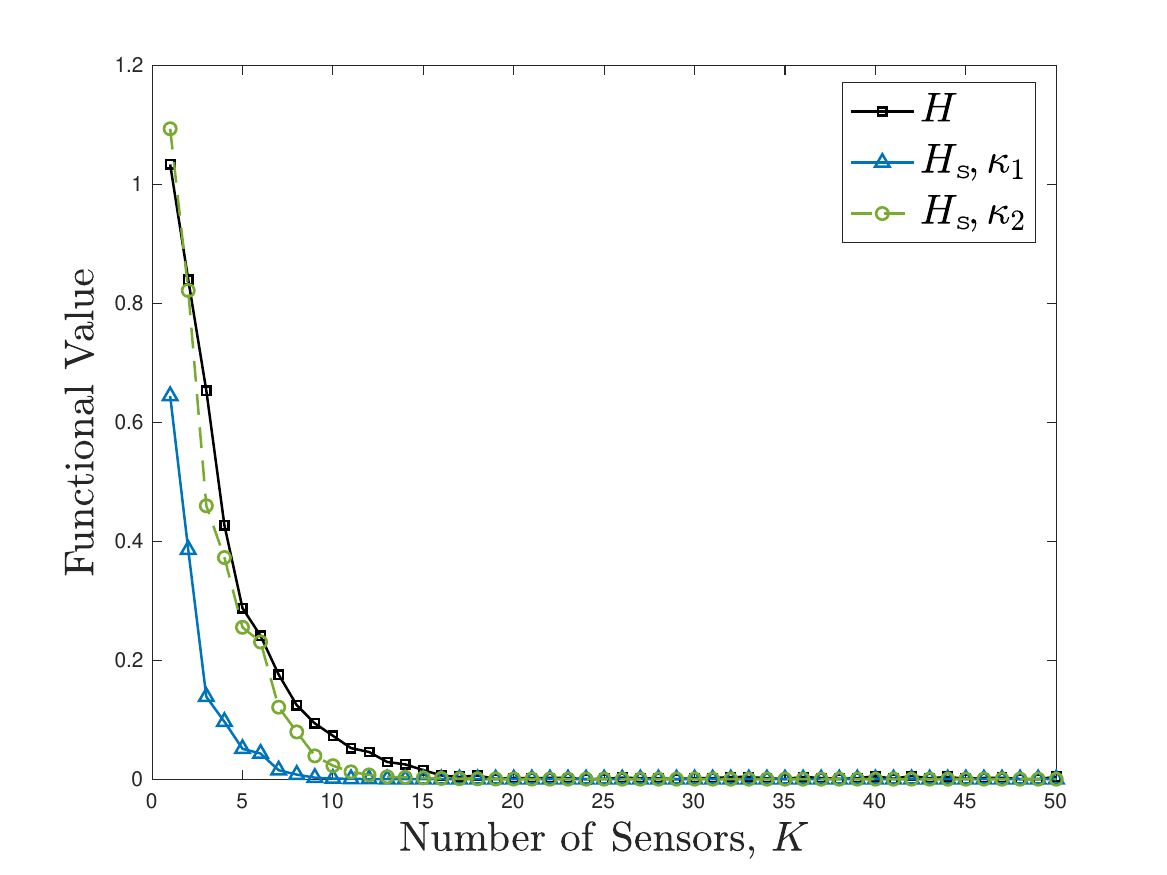}}
    \caption{Numerical validation on the surrogate sensing uncertainty, $\kappa_1 = \frac{1}{0.5M+1}$, $\kappa_2 =  \frac{1}{M+1}$.}
    \label{fig:approximation}
\end{figure}}
\end{Example}

\subsection{Aggregation Gain}
\subsubsection{Simplifying Uncertainty Surrogate}
To facilitate the analysis of multi-view aggregation gain, we simplify the expression of  uncertainty surrogate function in~\eqref{eq:surrogate} by Taylor expansion. First, define the average class separation distance as
\begin{equation}\label{eq:ave_separation}
    \bar{D} = \frac{1}{L(L-1)}\sum_{\ell=1}^L\sum_{\ell^{\prime}\neq \ell} D_{\ell,\ell^{\prime}},
\end{equation}
where $D_{\ell,\ell^{\prime}}$ is given in Lemma~\ref{Lemma:global_gain}. 
\begin{Proposition}[Surrogate Function Expansion]\label{Prop:surrogate}
    \emph{In the absence of channel noise, the surrogate function of sensing uncertainty in~\eqref{eq:surrogate} can be written as
    \begin{equation*}
        H_{\mathsf{s}}=\log \left[1+(L-1)\exp\left(-\kappa \bar{D}K\right)\right]  + C_b,
    \end{equation*}
    where $C_b = \mathcal{O}\left(\frac{1}{L(L-1)}\sum_{\ell}\sum_{\ell\neq\ell^{\prime}}(D_{\ell,\ell^{\prime}}-\bar{D})^2\right)$.}
\end{Proposition}
\begin{proof}
    (See Appendix~\ref{Apdx:relaxed_up}).
\end{proof}

In Proposition~\ref{Prop:surrogate}, the residual term $C_b$ is negligible when  between-class differentiability is similar, i.e., 
\begin{equation}\label{eq:appromation_condition}
    D_{\ell,\ell^{\prime}} \approx  D_{l,l^{\prime}},\ \forall \ell,\ell^{\prime},l,l^{\prime}. 
\end{equation}
To simplify notation, we assume that this is the case and hence  $C_b\approx 0$ in the subsequent analysis. As a result, the sensing uncertainty surrogate  reduces to 
\begin{equation}\label{eq:app_surrogate}
    H_{\mathsf{s}} \approx \log \left[1+(L-1)\exp\left(-\kappa \bar{D}K\right)\right].
\end{equation}
It is worth mentioning that our following analysis also holds for the case of $C_b\neq 0$ which, however, complicates  notation and makes analysis tedious. 

Based on ~\eqref{eq:app_surrogate}, $H_{\mathsf{s}}$ is observed to be a monotonically decreasing function of the product of the views' number and the average differentiability, say $K\cdot\bar{D}$. The product reflects two aspects of multi-view aggregation gain. On one hand, as $K$ grows, the aggregation over more views suppresses the variances of sensing data clusters and thereby decreases the sensing uncertainty. On the other hand, a growing number of suitably scheduled sensors also enhances the average class differentiability represented by $\bar{D}$ as elaborated shortly.

\subsubsection{Characterizing Average Class Separation Distance}
By substituting Lemma~\ref{Lemma:global_gain} into~\eqref{eq:ave_separation},  $\bar{D}$  can be written as
\begin{equation}\label{eq:average_gain_1}
    \begin{aligned}
         \bar{D} = \mathsf{Tr}(\bar{\mathbf{P}}\mathbf{C}^{-1}\bar{\mathbf{P}}\mathbf{D}),
    \end{aligned}
\end{equation}
where $\mathbf{D} = \frac{1}{L(L-1)}\sum_{\ell=1}^{L}\sum_{\ell^{\prime}\neq \ell} \left(\bmu_{\ell}-\bmu_{\ell^{\prime}}\right)\left(\bmu_{\ell}-\bmu_{\ell^{\prime}}\right)^{\top}$ reflects the average of pairwise class separation matrices, $\{\left(\bmu_{\ell}-\bmu_{\ell^{\prime}}\right)\left(\bmu_{\ell}-\bmu_{\ell^{\prime}}\right)^{\top}\}$. Therefore, $\bar{D}$ measures the components of $\mathbf{D}$ projected onto the subspace spanned by $\bar{\mathbf{P}}$. The  global  (multi-view) observation matrix $\bar{\mathbf{P}}$ has a larger rank than each of its local components  $\{\mathbf{P}_k\}$. However, including more views/sensors does not necessarily increase the value $\bar{D}$. It may even decrease by the addition of a sensor contributing little useful information in its observation sub-space, as observed in  Example 1. This suggests the need of designing a sensor scheduler  using the criterion of maximizing $\bar{D}$ when the views are limited.  

\subsubsection{Main Result}
With sufficient independent views in aggregation, the average of random observation matrices $\bar{\mathbf{P}}$  converges to its expectation, and thus the trace value $\mathsf{Tr}(\bar{\mathbf{P}}\mathbf{C}^{-1}\bar{\mathbf{P}}\mathbf{D})$ in~\eqref{eq:average_gain_1} reduces into a constant given as: 
\begin{equation}\label{eq:converged_const}
     \xi\overset{\triangle}{=}\mathsf{Tr}(\mathsf{E}[\mathbf{P}_k]\mathbf{C}^{-1}\mathsf{E}[\mathbf{P}_k]\mathbf{D}).
\end{equation}
Then the sensing uncertainty $H_{\mathsf{s}}$ in~\eqref{eq:app_surrogate} converges to 
\begin{equation}
H_{\mathsf{s}}=\log \left[1+(L-1)\exp\left(-\kappa \xi K\right)\right]. 
\end{equation}
Then the main result of this section is obtained as follows. 
\begin{MR}[View Aggregation Gain]\label{MR:asymptotic}
    \emph{For a large number of views, the sensing uncertainty with noiseless multi-view aggregation is exponentially decreasing w.r.t. $K$: 
    \begin{equation}\label{eq:mian_result_1}
       \boxed{\begin{aligned}
            H_{\mathsf{s}} &\approx (L-1)\exp(-\kappa \xi K), \quad\quad K\gg 1,
       \end{aligned}}
    \end{equation}
    where $\kappa$ can be $\frac{1}{2}$ and $\frac{1}{cM+2}$ according to the definition in~\eqref{eq:surrogate} and $\xi$ follows~\eqref{eq:converged_const}.}
\end{MR}
%
\section{Multi-view Aggregation Gain
with Channel Distortion}\label{section:analysis_2}
In this section, we build on the results in the preceding section to quantify the view-and-channel aggregation gain for the  E2E sensing performance of an ISEA system with AirComp over fading channels. In particular, the results on the distribution of aggregated features and  surrogate functions for sensing uncertainty are extended to account for channel distortion. We further obtain useful results on receive SNR distribution and channel-induced sensing-accuracy loss. 

\subsection{Aggregated Feature Distribution}
First, as the transmitted  feature maps are real, the real part of the aggregated feature map $\tilde{\mathbf{f}}$ is extracted as $\tilde{\mathbf{f}}^{\mathsf{re}} = \Re\{\tilde{\mathbf{f}}\}$. The distribution of $\tilde{\mathbf{f}}^{\mathsf{re}}$ is derived as shown below.
\begin{Lemma}\label{Lemma:aircomp_feature}
    \emph{Given AirComp in~\eqref{eq:post_process}, the aggregated feature map follows the Gaussian-mixture distribution: 
    \begin{equation*}
    \tilde{\mathbf{f}}^{\mathsf{re}}\sim\frac{1}{L}\sum_{\ell=1}^{L}\mathcal{N}\left(\bar{\mathbf{P}}\bmu_{\ell},\frac{1}{K}\mathbf{C} + \frac{1}{\gamma_{\mathsf{air}}}\mathbf{I}_M\right),
    \end{equation*}
    where $\gamma_{\mathsf{air}}\overset{\triangle}{=}\frac{2K^2}{\sigma^2\Vert\mathbf{b}\Vert^2}$ denotes the \emph{effective receive SNR}.}
\end{Lemma}
\begin{proof}
    (See Appendix~\ref{Apdx:aircomp_feature}).
\end{proof}

It follows from Lemma~\ref{Lemma:aircomp_feature} and~\eqref{eq:ML_classifier} that the optimal ML classifier in the current case is given as
\begin{equation}\label{eq:revised_classifier}
\ell^{\star}=\arg\min_{\ell}\ (\tilde{\mathbf{f}}^{\mathsf{re}}-\bar{\mathbf{P}}\bmu_{\ell})^{\top}\left(\frac{1}{K}\mathbf{C} + \frac{1}{\gamma_{\mathsf{air}}}\mathbf{I}_M\right)^{-1}(\tilde{\mathbf{f}}^{\mathsf{re}}-\bar{\mathbf{P}}\bmu_{\ell}).
\end{equation}

\subsection{Surrogate Function for Sensing Uncertainty}
Given the similar forms of aggregated feature distributions in Lemmas~\ref{Lemma:pooled_feature} and~\ref{Lemma:aircomp_feature}, the needed surrogate function can be derived similarly as its noiseless counterpart in Section~\ref{section:analysis_1}. First, based on Lemma~\ref{Lemma:global_gain} and~\ref{Lemma:aircomp_feature}, the global discrimination gain of $\tilde{\mathbf{f}}^{\mathsf{re}}$ is given as $\tilde{G}_{\ell,\ell^{\prime}}= K\tilde{D}_{\ell,\ell^{\prime}}(\gamma_{\mathsf{air}})$ where  the average class separation distance with channel distortion is given as
\begin{equation}\label{eq:revised_gain}
\tilde{D}_{\ell,\ell^{\prime}}(\gamma_{\mathsf{air}}) = \left(\bmu_{\ell}-\bmu_{\ell^{\prime}}\right)^{\top}\bar{\mathbf{P}}(\mathbf{C} + \frac{K}{\gamma_{\mathsf{air}}}\mathbf{I}_M)^{-1}\bar{\mathbf{P}}\left(\bmu_{\ell}-\bmu_{\ell^{\prime}}\right).
\end{equation}
One can observe that the effect of channel distortion  on the discrimination gain is regulated by a single parameter -- the effective receive SNR $\gamma_{\mathsf{air}}$. Using the above result and following the procedure as deriving Proposition~\ref{Prop:uncertainty}, the sensing uncertainty with channel distortion can be bounded as follows. 

\begin{Corollary}[Sensing Uncertainty with Channel Distortion]\label{Corollary:aircomp_uncertainty}
    \emph{Consider the ISEA system with linear classification and AirComp.  The resulting sensing uncertainty can be  bounded as
    \begin{equation*}
    \begin{aligned}
        &\frac{1}{L}\sum_{\ell=1}^{L}\log\left[1+\sum_{\ell^{\prime}\neq\ell}\exp\left(-\frac{\tilde{D}_{\ell,\ell^{\prime}}(\gamma_{\mathsf{air}})}{2}K\right)\right]\leq H(\gamma_{\mathsf{air}})\leq\\
        &\frac{1}{L}\sum_{\ell=1}^{L}\log\left[1+\sum_{\ell^{\prime}\neq\ell}\exp\left(-\frac{\tilde{D}_{\ell,\ell^{\prime}}(\gamma_{\mathsf{air}})}{cM+2}K\right)\right] +  C_a,
    \end{aligned}
    \end{equation*}
    where $c>0$, $C_a$ is the constant same as in Proposition~\ref{Prop:uncertainty}, and $\tilde{D}_{\ell,\ell^{\prime}}(\gamma_{\mathsf{air}})$ is given in~\eqref{eq:revised_gain}.}
\end{Corollary}
It follows that the noisy counterpart of the uncertainty  surrogate function in~\eqref{eq:app_surrogate} can be obtained as
\begin{equation}\label{eq:aircomp_uncertainty}
    H_{\mathsf{s}}(\gamma_{\mathsf{air}})= \log \left[1+(L-1)\exp\left(-\kappa \tilde{D}(\gamma_{\mathsf{air}}) K\right)\right],
\end{equation}
    where $ \tilde{D}(\gamma_{\mathsf{air}}) \overset{\triangle}{=}\frac{1}{L(L-1)}\sum_{\ell=1}^L\sum_{\ell^{\prime}\neq \ell} \tilde{D}_{\ell,\ell^{\prime}}(\gamma_{\mathsf{air}})= \mathsf{Tr}(\bar{\mathbf{P}}(\mathbf{C} + \frac{K}{\gamma_{\mathsf{air}}}\mathbf{I}_M)^{-1}\bar{\mathbf{P}}\mathbf{D})$.

\subsection{Distribution of Effective Receive SNR}
The dependence of sensing uncertainty on the effective receive SNR, $\gamma_{\mathsf{air}}$, as reflected in~\eqref{eq:aircomp_uncertainty} suggests the need of analyzing its distribution, which is carried out as follows.

\begin{Lemma}\label{Lemma:monotonical}
    \emph{The sensing uncertainty function $H_{\mathsf{s}}(\gamma_{\mathsf{air}})$ in~\eqref{eq:aircomp_uncertainty} is monotonically decreasing w.r.t. $\gamma_{\mathsf{air}}$.}
\end{Lemma}
\begin{proof}
    (See Appendix~\ref{Apdx:monotonical}).
\end{proof}
Under the power constraint $\frac{\nu^2}{\mathbf{b}^H\mathbf{h}_k\mathbf{h}_k^H\mathbf{b}}\leq P$, the effective receive SNR maximized by optimal beamforming is given as~\cite{GXZhuAirComp2019}
\begin{equation}\label{eq:opt_norm}
    \gamma_{\mathsf{air}} = \frac{2K^2\gamma}{\nu^2}\cdot \min_k\ \mathbf{v}^H\mathbf{h}_k\mathbf{h}_k^H\mathbf{v},
\end{equation}
where $\mathbf{v}$ denotes the optimal receive beamformer computed as the first eigenvector of the channel matrix $\mathbf{H}=[\mathbf{h}_1,\mathbf{h}_2,\cdots,\mathbf{h}_K]$. As revealed in~\eqref{eq:opt_norm}, $\gamma_{\mathsf{air}}$ is limited by the weakest link due to signal magnitude alignment of AirComp. The alignment term, $\min_k\ \mathbf{v}^H\mathbf{h}_k\mathbf{h}_k^H\mathbf{v}$, will reduce to zero as $K\rightarrow \infty$. However, its value scaled up proportional to $K$, say $\zeta_{\mathsf{air}} \overset{\triangle}{=}K\cdot\min_k\ \mathbf{v}^H\mathbf{h}_k\mathbf{h}_k^H\mathbf{v}$, can converge to an exponential random variable with a fixed parameter. This gives the distribution of $\gamma_{\mathsf{air}}$ as follows.

\begin{Lemma}[Asymptotic Distribution of Effective Receive SNR]\label{Lemma:asymp_dist}
    \emph{For a large number of sensors ($K\rightarrow \infty$) and a proportional array size $N = \omega K$ with $\omega\in(0,\infty)$, the effective receive SNR resulting from AirComp $\gamma_{\mathsf{air}} = 2K\frac{\gamma}{\nu^2}\cdot\zeta_{\mathsf{air}}$ with $\zeta_{\mathsf{air}}$ being an exponential random variable:
    \begin{equation*}
        \zeta_{\mathsf{air}}\sim \mathrm{Exp}\left(\frac{1}{(1+\sqrt{\omega})^2}\right).
    \end{equation*}}
\end{Lemma}
\begin{proof}
    (See Appendix~\ref{Apdx:asymp_dist}).
\end{proof}
\begin{figure}[t]
    \centering
    \subfigure[$K=100$]{\label{subfig:K100}\includegraphics[width=0.44\textwidth]{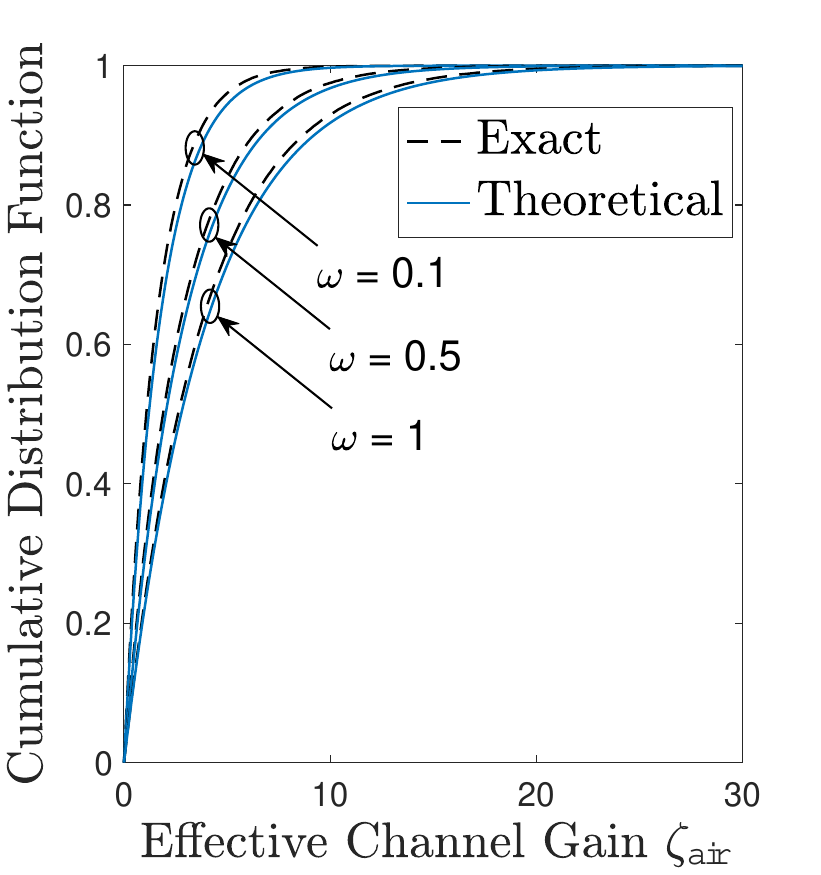}}
    \subfigure[$K=200$]{\label{subfig:K200}\includegraphics[width=0.44\textwidth]{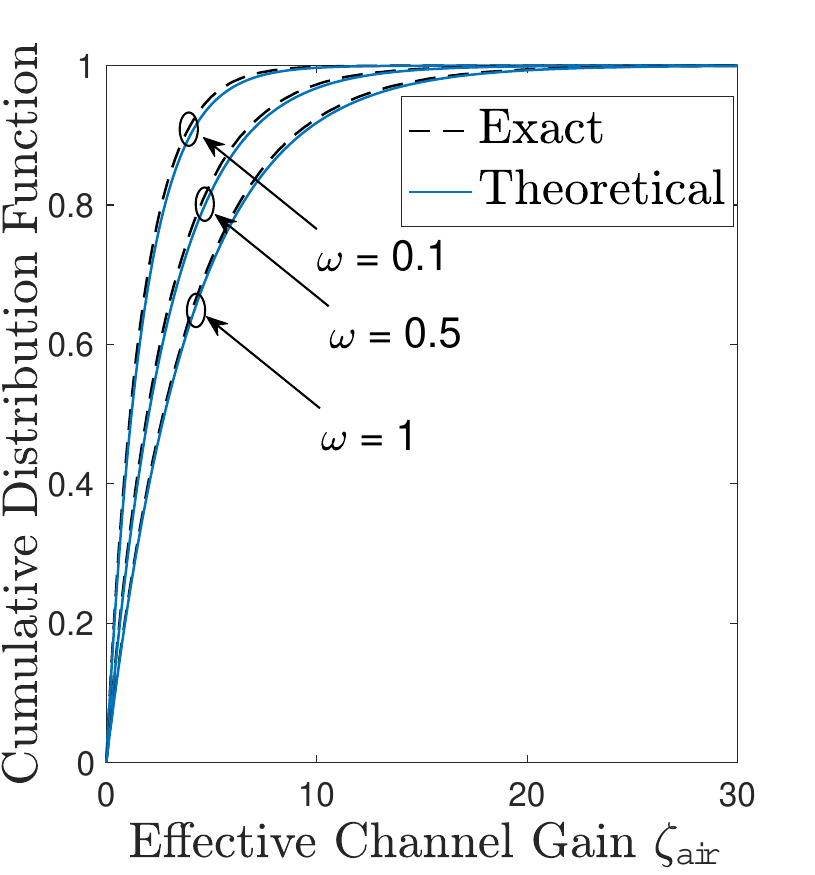}}
    \caption{Numerical validation of Lemma~\ref{Lemma:asymp_dist}.}
    \label{fig:asymp_dist}
\end{figure}
The asymptotic distribution of $\zeta_{\mathsf{air}}$ is numerically validated in Fig.~\ref{fig:asymp_dist}. The distribution can be further extended to the case with a fixed array size, in which the continuously increasing number of sensors will lead to $\omega = 0$ and $\zeta_{\mathsf{air}}\sim \mathrm{Exp}\left(1\right)$. Using Lemma~\ref{Lemma:asymp_dist}, the scaling property of exponential distributions yields the asymptotic distribution of the receive SNR:
\begin{equation}\label{eq:eff_dist_SNR}
   \frac{\gamma_{\mathsf{air}}}{K}\sim\mathrm{Exp}\left(\frac{\nu^2}{2\gamma(1+\sqrt{\omega})^2}\right),\quad\ K\rightarrow \infty, N=\omega K. 
\end{equation}
\begin{Remark}[Channel Noise Suppression]
    \emph{It follows from~\eqref{eq:eff_dist_SNR} that the power of channel noise in AirComp, given by $\frac{1}{\gamma_{\mathsf{air}}} = \frac{1}{\gamma_{\mathsf{air}}/K}\cdot\frac{1}{K}$, is inversely proportional to $K$ as the term $\gamma_{\mathsf{air}}/K$ is independent of $K$.}
\end{Remark}

\subsection{Aggregation Gain}
Given the preceding analysis, we are ready to quantify the view-and-channel aggregation gain under the joint effects of view aggregation that suppresses sensing noise and channel aggregation that suppresses channel noise. The core of the analysis is to derive a key variable -- channel-induced loss on sensing performance as follows.
\subsubsection{Channel Induced Performance}
As shown in~\eqref{eq:aircomp_uncertainty}, multi-view aggregation with channel distortion can achieve sensing uncertainty with the same form as its noiseless counterpart, except for reducing the average class differentiability from $\bar{D}$ to $\tilde{D}(\gamma_{\mathsf{air}})$ defined in~\eqref{eq:aircomp_uncertainty}. It follows that their ratio can represent the channel-induced performance loss: $A_{\mathsf{loss}} \overset{\triangle}{=} \frac{\tilde{D}(\gamma_{\mathsf{air}})}{\bar{D}}$. By using the Woodbury matrix identity~\cite{MatrixAnalysis}, $\tilde{D}(\gamma_{\mathsf{air}})$ in~\eqref{eq:aircomp_uncertainty} can be expressed as
\begin{align}\label{eq:distance_loss}
        &\tilde{D}(\gamma_{\mathsf{air}}) \nonumber\\
        &= \mathsf{Tr}\left(\bar{\mathbf{P}}\mathbf{C}^{-1}\bar{\mathbf{P}}\mathbf{D}\right) - \mathsf{Tr}(\bar{\mathbf{P}}\mathbf{C}^{-1}(\mathbf{C}^{-1} + \frac{K}{\gamma_{\mathsf{air}}}\mathbf{I}_M)^{-1}\mathbf{C}^{-1}\bar{\mathbf{P}}\mathbf{D}),\nonumber\\
        & = \bar{D} - \mathsf{Tr}(\bar{\mathbf{P}}\mathbf{C}^{-1}(\mathbf{C}^{-1} + \frac{K}{\gamma_{\mathsf{air}}}\mathbf{I}_M)^{-1}\mathbf{C}^{-1}\bar{\mathbf{P}}\mathbf{D}).
\end{align}
As a result,
\begin{equation}\label{eq:loss_of_gain}
        A_{\mathsf{loss}}= 1 - \frac{\mathsf{Tr}(\bar{\mathbf{P}}\mathbf{C}^{-1}(\mathbf{C}^{-1} + \frac{K}{\gamma_{\mathsf{air}}}\mathbf{I}_M)^{-1}\mathbf{C}^{-1}\bar{\mathbf{P}}\mathbf{D})}{\mathsf{Tr}\left(\bar{\mathbf{P}}\mathbf{C}^{-1}\bar{\mathbf{P}}\mathbf{D}\right)}.
\end{equation}
\begin{equation}
H_{\mathsf{s}}(\gamma_{\mathsf{air}})=\log \left[1+(L-1)\exp\left(-\kappa\bar{D}KA_{\mathsf{loss}}\right)\right].
\end{equation}
\subsubsection{Main Result}
Based on~\eqref{eq:eff_dist_SNR}, $A_{\mathsf{loss}}$ is independent of $K$ as $K\rightarrow \infty$. Therefore, combining~\eqref{eq:aircomp_uncertainty} and~\eqref{eq:loss_of_gain} yields the main result. 
\begin{MR}[View-and-Channel Aggregation Gain]\label{MR:mian_result_2}
     \emph{Consider an ISEA system employing AirComp-based multi-view aggregation, the sensing uncertainty is exponentially decreasing w.r.t. $K$:
    \begin{equation}\label{eq:mian_result_2}
       \boxed{H_{\mathsf{s}}(\gamma_{\mathsf{air}})\approx(L-1)\exp\left(-\kappa\xi A_{\mathsf{loss}} K\right),\quad\quad K\gg 1,}
    \end{equation}
where $\kappa$ and $\xi$ are given in~\eqref{eq:mian_result_1}.}
\end{MR}
The above result shows that due to channel aggregation in AirComp, channel distortion does not change the exponential decay of sensing uncertainty but does reduce the exponential rate by a factor of $A_{\mathsf{loss}}$.
 
Last, we investigate the effects of several system parameters on the key variable $A_{\mathsf{loss}}$. Define $r=2\gamma\nu^{-2}(1+\sqrt{\omega})^2\lambda_{\mathbf{C},\min}$ for ease of notation, where $\gamma$, $\nu^2$, $\omega$ and $\lambda_{\mathbf{C},\min}$ denote transmit SNR, the variance of transmit symbols, sensor-antenna number ratio, and the largest eigenvalue of the feature covariance matrix $\mathbf{C}$, respectively. Then, using the distribution of $\gamma_{\mathsf{air}}$ in~\eqref{eq:eff_dist_SNR}, the expectation of $A_{\mathsf{loss}}$ is obtained as
\begin{align}\label{eq:mean_loss}
    \mathsf{E}_{\frac{K}{\gamma_{\mathsf{air}}}}\left[A_{\mathsf{loss}}\right]
    &\geq 1-\mathsf{E}_{\frac{K}{\gamma_{\mathsf{air}}}}\left[\frac{1}{1+\frac{K}{\gamma_{\mathsf{air}}}\lambda_{\mathbf{C},\min}}\right]\nonumber\\
    & \overset{(a)}{=}1 -\frac{\exp(r^{-1})}{r} \cdot\mathsf{E}_1\left(r^{-1}\right)\nonumber\\
    &\overset{(b)}{\geq}1-\frac{\ln \left(1+r\right)}{r},
\end{align}
where the inequality follows from the trace inequalities~\cite{MatrixAnalysis}, the $\mathsf{E}_1\left(x\right)$ in step (a) denotes the exponential integral defined as $\mathsf{E}_1\left(x\right) = \int_x^\infty \frac{e^{-t}}{t}\mathrm{d}t$, and the inequality $\mathsf{E}_1\left(x\right)\leq e^{-x}\ln(1+1/x)$ is used in step (b). The effects of the system parameters on $A_{\mathsf{loss}}$ are then inferred from~\eqref{eq:mean_loss}. Specifically, transmit SNR $\gamma$ can linearly enlarge $r$ and thereby increases $A_{\mathsf{loss}}$ with the scaling of $\frac{\ln \left(1+\gamma\right)}{\gamma}$. In view aggregation, letting the number of antennas up faster than $K$ can give a large value of $\omega=N/K$. If $\omega \gg1$, $(1+\sqrt{\omega})^2\approx \omega$, $A_{\mathsf{loss}}$ increases w.r.t. $\omega$ at the rate of $\frac{\ln \left(1+\omega\right)}{\omega}$.
%
\section{AirComp or Analog Orthogonal Access?}\label{section:comparison}
With limited receive antennas ($N\leq K$), AirComp enables simultaneous access while (spatial division) orthogonal access is infeasible. On the other hand, with a large receive array, both schemes are feasible. In this section, we benchmark AirComp against analog orthogonal access. The study leads to the development of a new scheme supporting dynamic access-mode switching.

\subsection{Performance of Analog Orthogonal Access}
The effective feature map received by using analog orthogonal access can be extracted from the real part of $\tilde{\mathbf{f}}$ in~\eqref{eq:OA_analog}: $\tilde{\mathbf{f}}^{\mathsf{re}}=\Re\left\{\tilde{\mathbf{f}}\right\}= \bar{\mathbf{f}} + \frac{\nu}{K\sqrt{P}}\Re\left\{\mathbf{Z}^H\sum_k\mathbf{b}_k\right\}$. It shows that the analog orthogonal access introduces Gaussian channel noise into the ground-truth features with the covariance of $\frac{1}{\gamma_{\mathrm{aoa}}}\mathbf{I}_M$  and $\gamma_{\mathrm{aoa}}\overset{\triangle}{=}\frac{\gamma K^2}{\nu^2\sum_k\Vert\mathbf{b}_k\Vert^2}$ being the effective receive SNR. The result is similar to the case of AirComp except for the changed effective SNR. This allows the sensing uncertainty in~\eqref{eq:aircomp_uncertainty} to be modified for analog orthogonal access as
\begin{equation}\label{eq:orthogonal_uncertianty}
    H_{\mathsf{s}}(\gamma_{\mathrm{aoa}})= \log \left[1+(L-1)\exp\left(-\kappa \tilde{D}(\gamma_{\mathrm{aoa}})K)\right)\right],
\end{equation}
where $ \tilde{D}(\gamma_{\mathrm{aoa}}) = \mathsf{Tr}\left(\bar{\mathbf{P}}\left(\mathbf{C} + \frac{K}{\gamma_{\mathrm{aoa}}}\mathbf{I}_M\right)^{-1}\bar{\mathbf{P}}\mathbf{D}\right)$. 
\subsection{Crossing Point and Access Mode Switching}
Comparing the uncertainty functions in~\eqref{eq:orthogonal_uncertianty} and~\eqref{eq:aircomp_uncertainty} and using their monotonicity (see Lemma~\ref{Lemma:monotonical}), we can infer that AirComp outperforms analog orthogonal access if $ \gamma_{\mathrm{air}}\geq\gamma_{\mathrm{aoa}}$, and vice versa. This motivates us to propose the scheme of adaptive access-mode switching, referring to as \emph{adaptive access}, as 
    \begin{equation}\label{eq:adaptive}
        \begin{aligned}
            \text{Multi-access\ mode}=\left\{ \begin{array}{ll}
                  \text{AirComp},& \!\! \gamma_{\mathrm{air}}\geq \gamma_{\mathrm{aoa}},  \\
                  \text{Anal.\ orthog.\ access},&\!\!  \gamma_{\mathrm{air}}< \gamma_{\mathrm{aoa}}.
            \end{array}\right.
        \end{aligned}
    \end{equation}
Given adaptive access, the dependence of optimal access mode on system parameters $N$ and $K$ is understood in the sequel. To this end, a useful result on the distribution of the square norm of the ZF beamformers in analog orthogonal access, $\Vert\mathbf{b}_k\Vert^2$, is derived as follows.
\begin{Lemma}\label{Lemma:independence}
    \emph{Given $\mathbf{b}_k=\mathbf{H}(\mathbf{H}^H\mathbf{H})^{-1}\mathbf{e}_k$ in analog orthogonal access, $\Vert\mathbf{b}_k\Vert^2$ follows an i.i.d. distribution of $\Vert\mathbf{b}_k\Vert^2\sim 2\cdot\mathrm{Inv-}\chi_{2(N-K+1)}^2$, where $\mathrm{Inv-}\chi_{2(N-K+1)}^2$ denotes the inverse chi-square distribution with $2(N-K+1)$ degrees of freedom.}
\end{Lemma}
\begin{proof}
    (See Appendix~\ref{Apdx:independence}).
\end{proof}
Using Lemma~\ref{Lemma:independence} and the law of large numbers, the effective receive SNR for analog orthogonal access can be approximated as
\begin{equation}\label{eq:app_aoa}
\begin{aligned}
        \frac{\gamma_{\mathrm{aoa}}}{K} = \frac{\gamma}{\nu^2\frac{1}{K}\sum_k\Vert\mathbf{b}_k\Vert^2}\approx \frac{\gamma(N-K)}{\nu^2},\quad\quad K\gg 1.
\end{aligned}
\end{equation}
It follows from~\eqref{eq:eff_dist_SNR} and~\eqref{eq:app_aoa} that if $N=K$, $\frac{\gamma_{\mathrm{aoa}}}{K} \approx 0$ and $\frac{\gamma_{\mathrm{air}}}{K}\geq0$, leading to the event,  $ \gamma_{\mathrm{air}}\geq\gamma_{\mathrm{aoa}}$, occurring with probability $1$. On the other hand, if $N>K$, the probability of this event is asymptotically close to $0$ since using~\eqref{eq:eff_dist_SNR} and~\eqref{eq:app_aoa}
\begin{align}
    \mathrm{Pr}\left(\gamma_{\mathrm{air}}\geq\gamma_{\mathrm{aoa}}\right) &= \mathrm{Pr}\left(\frac{\gamma_{\mathrm{air}}}{K}\geq\frac{K\gamma(\omega-1)}{\nu^2}\right)\nonumber\\
    &=\exp\left(-\frac{K}{2}\frac{\sqrt{\omega}-1}{\sqrt{\omega}+1}\right)\rightarrow 0,\ \mathrm{as\ }K\rightarrow \infty.
\end{align}
Combining the above results gives the following conclusion.
\begin{MR}[Mode Switching Point]\label{MR:asymp_comparison}
    \emph{There exists a crossing point of the sensing uncertainty between AirComp and analog orthogonal access. Given $N = \omega K$ and $K>>1$, AirComp outperforms analog orthogonal access for the case of $\omega\leq 1$ (i.e., $N\leq K$); otherwise, the reverse holds. In other words, the crossing point is around $N=K$. This conclusion is numerically validated in Fig.~\ref{fig:crossingpoint}.}
\end{MR}
\begin{figure}[t]
    \centering
    \subfigure[$K=10$]{\label{subfig:K10}\includegraphics[width=0.44\textwidth]{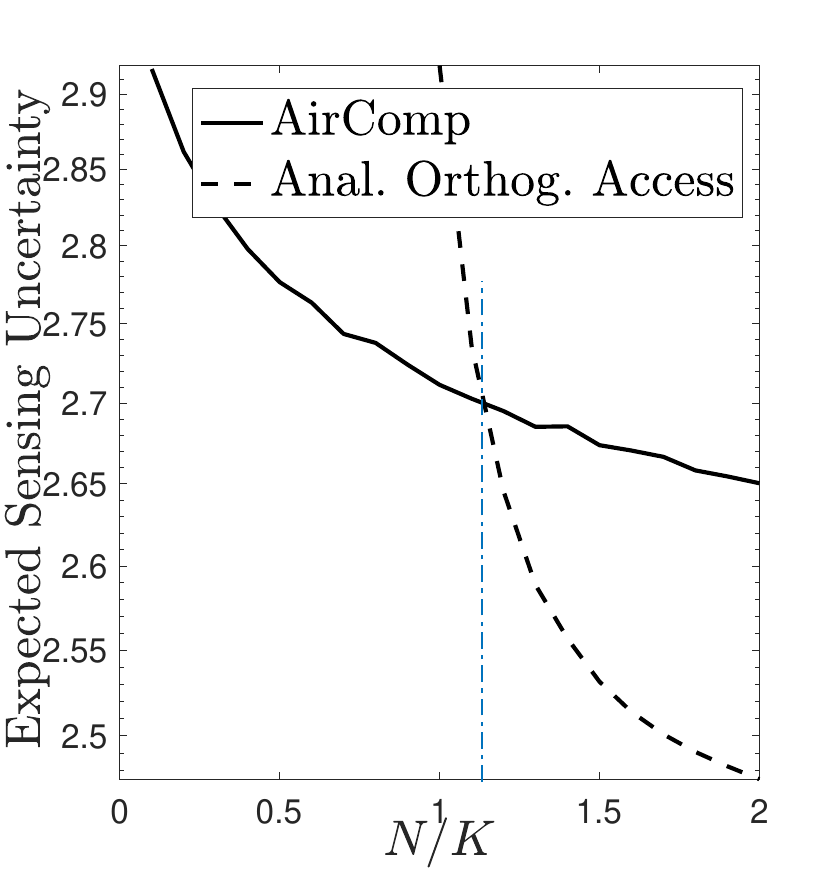}}
    \subfigure[$K=50$]{\label{subfig:K50}\includegraphics[width=0.44\textwidth]{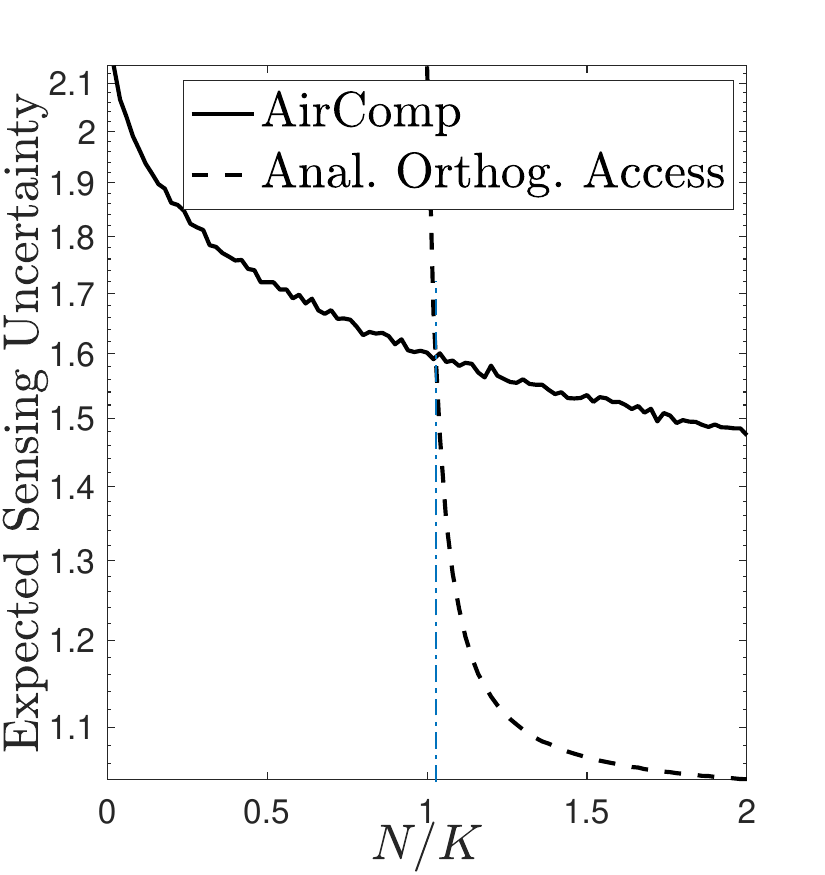}}
    \caption{Numerical validation on the crossing point between AirComp and analog orthogonal access. The parameters are set as $M=L=10$, $\mathbf{C} = 0.1\mathbf{I}_M$, $\mathsf{rank}(\mathbf{P}_k)=1$, $\gamma=10\mathrm{\ dB}$. The expectation is taken over channel distribution.}
    \label{fig:crossingpoint}
\end{figure}

%
\section{Experimental Results}\label{section:experiments}

\subsection{Experimental Settings}

Consider the ISEA system as shown in Fig.~\ref{fig:system_model}. Assuming frequency non-selective Rayleigh fading, the multi-access channel is composed of i.i.d. Gaussian $\mathcal{N}(0,1)$ elements. The coherence duration of the channel spreads over $256$ symbol slots, supporting analog transmission of feature vector with the maximum length of $256$. For the MVCNN architecture, we consider both the cases of linear classification on synthetic data and CNN-based classification on real-world data as follows.
\begin{itemize}
    \item \emph{Linear classification on synthetic data}: Local feature maps are drawn from the GMM in~\eqref{eq:local_PDF} and fed into the classifier given in~\eqref{eq:revised_classifier} after over-the-air averaging via AirComp. The feature maps' dimensionality is $M=100$, the number of classes is $L=20$, and the covariance matrix $\mathbf{C} = 0.1\mathbf{I}_M$. The observation matrices $\{\mathbf{P}_k\}$ are randomly generated as the principal eigenspaces of random matrices with i.i.d. Gaussian entries. For instance, let $\mathbf{G}$ be a randomly generated $M\times M$ Gaussian matrix. Then, $\mathbf{P}_k = \mathbf{U}_{\mathbf{G}}\mathbf{U}_{\mathbf{G}}^{\top}$ with $\mathbf{U}_{\mathbf{G}}$ being the $\mathsf{rank}(\mathbf{P}_k)$-dimensional principal eigenspace of $\mathbf{G}$.
    \item \emph{MVCNN-based classification on real-world data}: We consider the well-known \emph{ModelNet} dataset which comprises multi-view images of objects (e.g., sofas and tables) and the popular VGG11 model for implementing the MVCNN architecture. The VGG11 is split before the linear classifier with the classifier employed at the server and the other components deployed at each sensor for feature extraction~\cite{airpooling2023}. The resultant MVCNN architecture is trained for average pooling. Therein, we select a data subsect of ModelNet corresponding to $L=10$ popular object classes for our experiments. The data entries for each target are captured by $K=12$ sensors (i.e., cameras) with the angle between adjacent sensors being $30^{\circ}$. Each feature map output from an on-device model is described as a $512\times7\times7$ tensor, where the $512\times7$ slices of these feature tensors are transmitted and aggregated sequentially at the server for global classification. 
\end{itemize}

Last, to evaluate the performance of AirComp, we adopt two benchmarking schemes, namely analog orthogonal access in~\eqref{eq:OA_analog} and the adaptive (dual-mode) access in~\eqref{eq:adaptive}, to support local feature uploading.
\begin{figure}[t]
    \centering
    \includegraphics[width=0.8\textwidth]{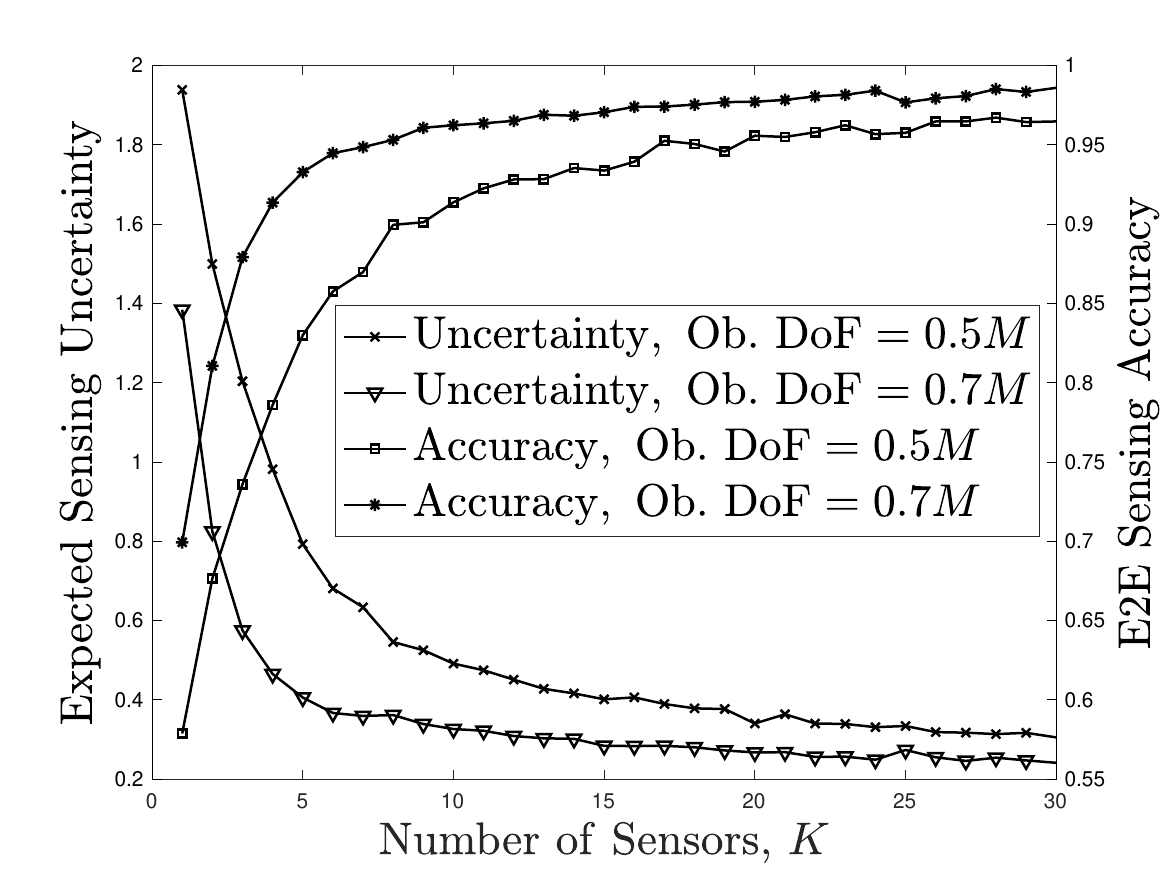} %
    \caption{(Linear classification) Comparison between E2E sensing uncertainty and accuracy for a variable number of sensors and different local observation DoFs.}
    \label{fig:unvertainty_accuracy}
\end{figure}
\begin{figure}[t]
    \centering
    \subfigure[]{\label{subfig:acc_N_1}\includegraphics[width=0.44\textwidth]{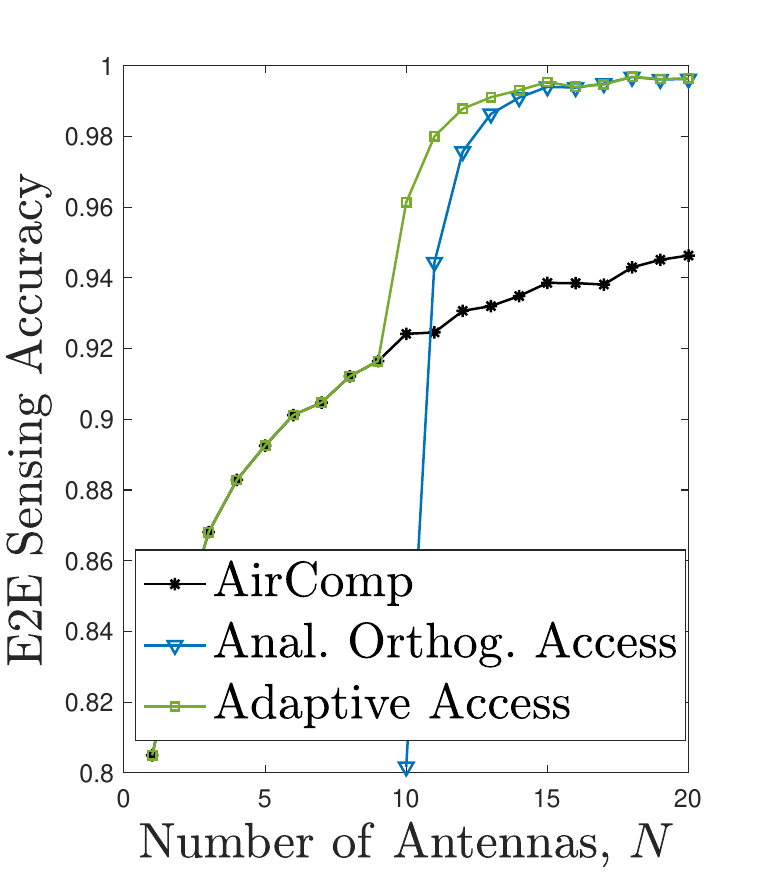}}
    \subfigure[]{\label{subfig:gain}\includegraphics[width=0.44\textwidth]{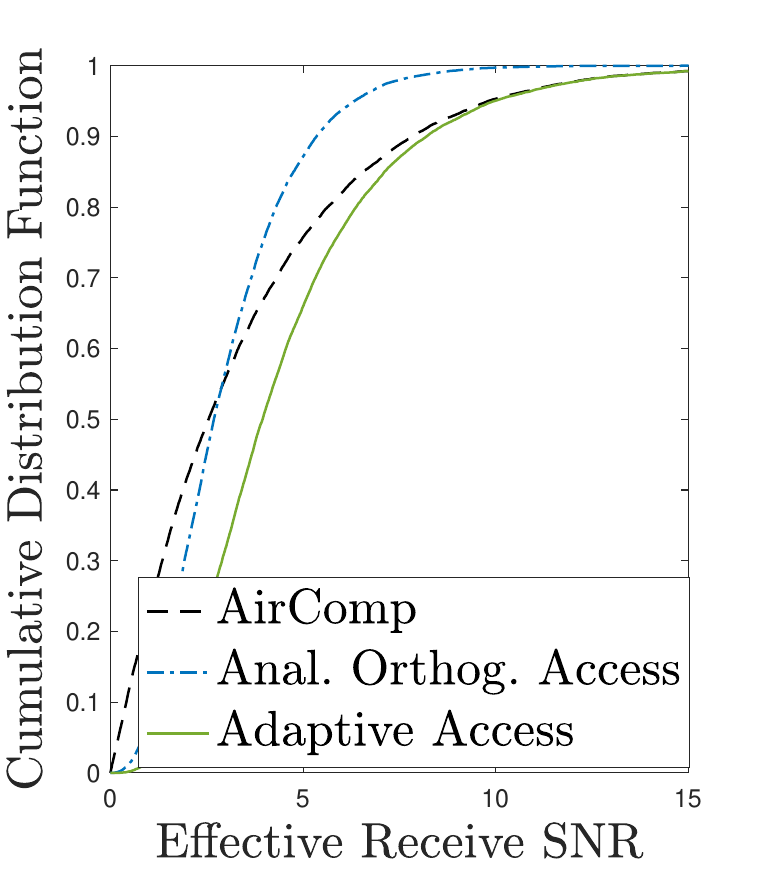}}
    \caption{(Linear classification) Performance comparison between AirComp, analog orthogonal access, and adaptive access in terms of (a) E2E sensing accuracy and (b) effective receive SNR with $N=12$.}
    \label{fig:accuracy_syn}
\end{figure}
\subsection{ISEA with Linear Classification}
The curves of E2E sensing accuracy and uncertainty versus number of sensors, $K$, are plotted in Fig.~\ref{fig:unvertainty_accuracy}. Different levels of local observation DoFs are considered: $\mathsf{rank}(\mathbf{P}_k)=\{0.5M,0.7M\}$ for all $k$. First, it can be observed from Fig.~\ref{fig:unvertainty_accuracy} that the E2E sensing uncertainty diminishes at an exponential rate as $K$ grows. This is consistent with the main analytical results in~\eqref{eq:mian_result_2}. On the other hand, the E2E sensing accuracy converges to the saturation level (i.e., maximum accuracy) also exponentially fast. The consistency of uncertainty and accuracy validates their duality. One can also observe the existence of a critical range (i.e., $K\leq 10$) where the sensing performance is sensitive to changes on sensor number. Last, increasing local observation DoFs is found to yield significant gains on sensing performance. 

In Fig.~\ref{fig:accuracy_syn}, the performance of ISEA using AirComp is compared with that of counterparts employing benchmarking access schemes. Specifically, Fig.~\ref{subfig:acc_N_1} depicts the curves of E2E sensing accuracy versus receive array size, $N$; the cumulative distribution function (CDF) curves of effective receive SNR are plotted in Fig.~\ref{subfig:gain}. The number of sensors is fixed as $K=10$. The most important observation can be made from Fig.~\ref{subfig:acc_N_1} that there exists a crossing point between AirComp and analog orthogonal access at around $N=K$. This validates the Main Result~\ref{MR:asymp_comparison}. Next, the adaptive access scheme designed in Section~\ref{section:comparison} is observed to be effective as it outperforms the other two underpinning schemes. The above observations are consistent with those from Fig.~\ref{subfig:gain}, i.e., the crossing point of SNR CDF curves and superiority of adaptive access in terms of effective receive SNR.
\begin{figure}[t]
    \centering
    \subfigure[]{\label{subfig:acc_K}\includegraphics[width=0.44\textwidth]{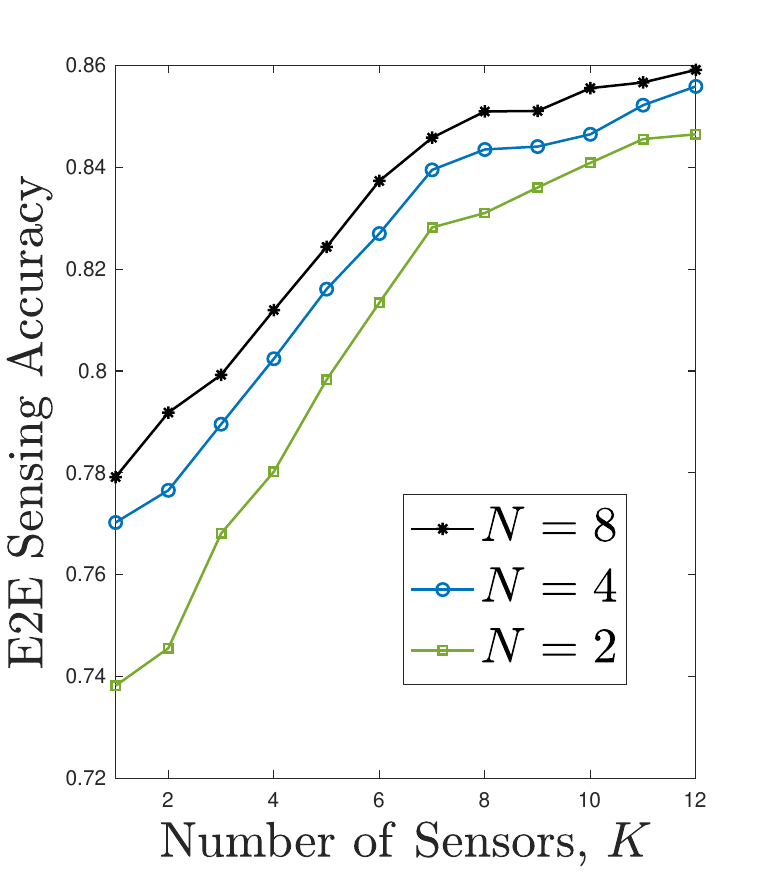}}
    \subfigure[]{\label{subfig:acc_N}\includegraphics[width=0.44\textwidth]{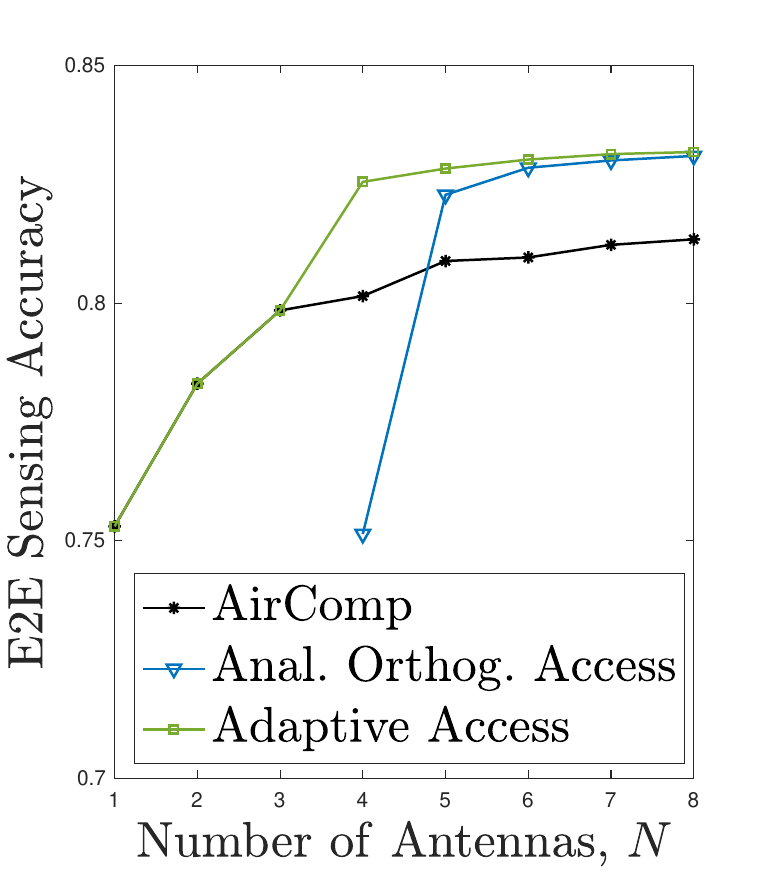}}
    \caption{(MVCNN classification) The dependence of E2E sensing accuracy on (a) the numbers of sensors and (b) receive antennas under transmit SNR $\gamma=-10$ dB.}
    \label{fig:accuracy_mvcnn}
\end{figure}

\begin{figure}[t]
    \centering
    \subfigure[$K = 4$]{\label{subfig:N4_K4}\includegraphics[width=0.44\textwidth]{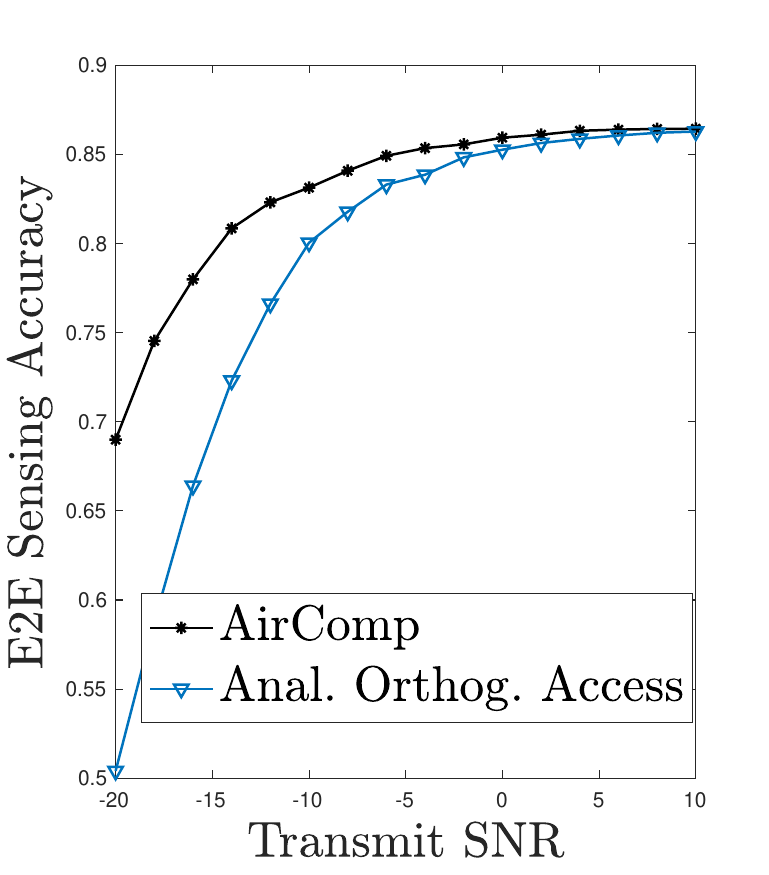}}
    \subfigure[$K = 8$]{\label{subfig:N8_K8}\includegraphics[width=0.44\textwidth]{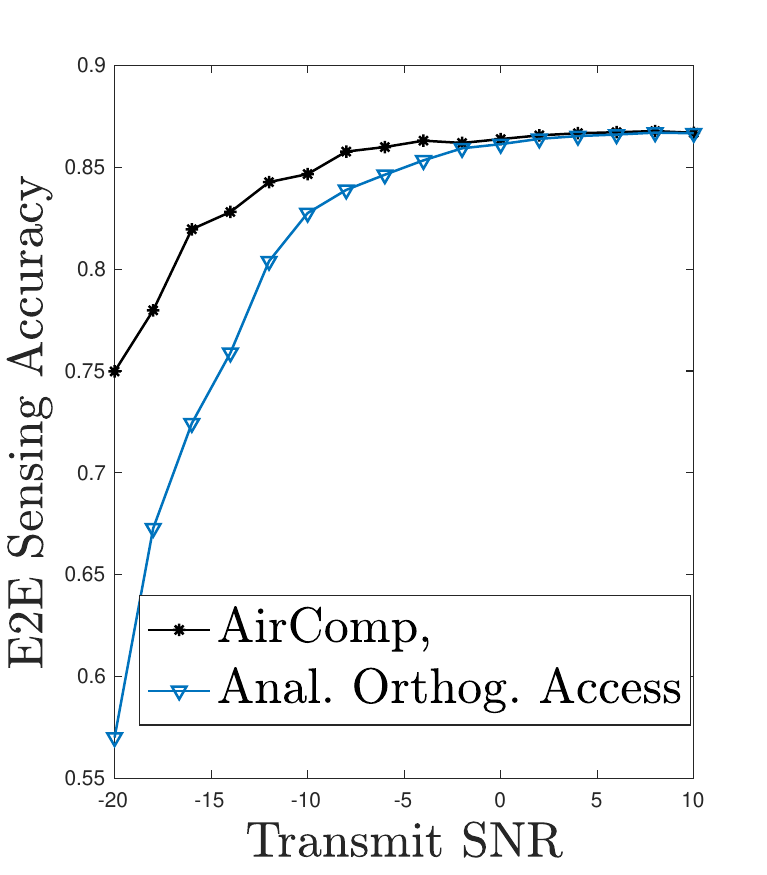}}
    \caption{(MVCNN classification) The effect of transmit SNR on the E2E sensing accuracy for the number of sensors $K=N=\{4,8\}$.}
    \label{fig:accuracy_SNR}
\end{figure}
\subsection{MVCNN-based Classification} 
Experimental results for the case of MVCNN based classification are presented to validate the insights from our analysis based on linear classification. Specifically, the MVCNN counterparts of the E2E sensing accuracy curves in Fig.~\ref{fig:unvertainty_accuracy} and those in Fig.~\ref{subfig:acc_N_1} are obtained as plotted in Fig.~\ref{subfig:acc_K} and Fig.~\ref{subfig:acc_N}, respectively. The main observations from the MVCNN curves are lowly identical to those from their linear-classification counterparts. Specifically, the exponential convergence of E2E sensing accuracy is reflected in Fig.~\ref{subfig:acc_K}; the crossing point between AirComp and analog orthogonal access is found in Fig.~\ref{subfig:acc_N} to be also around $N=K$.

Last, the curves of E2E sensing accuracy versus transmit SNR are plotted in Fig.~\ref{fig:accuracy_SNR} for different values of $(N,K)$. The main observation is that the close-to-maximum accuracy is achievable even at very low transmit SNR (e.g., $-5$ dB). As mentioned early, the reason is view-and-channel aggregation gain in two aspects. First, the aggregation gain enhances the receive SNR by a factor approximately equal to $K$. In other words, $K=10$ can achieve the effective receive SNR of $10$ dB even given transmit SNR as low as $0$ dB. The other aspect is that view aggregation improves the model's classification margin (see, e.g.,~\cite{airpooling2023}) to absorb channel distortion without compromising sensing accuracy. The above observation advocates the use of AirComp and uncoded analog transmission at large to support fast ISEA.

\section{Conclusion}
We have presented a theoretical framework for characterizing the performance gains from view-and-channel aggregation in an ISEA system. Our results reveal that the sensing/inference uncertainty decreases exponentially with the increasing number of views/sensors, with the rate being directly proportional to the global discriminant gain. Furthermore, it is demonstrated that the channel distortion resulting from aggregation via AirComp or analog orthogonal access does not alter this scaling law, except for a reduction in the exponential rate. Utilizing the end-to-end performance analysis, we have also developed a scheme for aggregation mode adaption that dynamically switches between AirComp and orthogonal analog access to achieve optimal system performance.

Expanding upon the insights from this study, we anticipate that distributed sensing, leveraging its multi-view aggregation gain, will emerge as a mainstream direction in the area of ISEA.
Advancing this direction calls for novel quantitative analysis and protocol designs that strike a balance between access control, latency, and computation accuracy so as to improve the end-to-end performance. This study serves as an initial step in establishing a theoretical framework for the advancement of ISEA. Extending this framework to incorporate other wireless techniques (such as broadband transmission and random access) and other sensing scenarios (including sensing via point clouds and multi-modal fusion) warrant further investigation.

\appendix
\subsection{Proof of Lemma~\ref{Lemma:pooled_feature}}\label{Apdx:pooled_feature}
Using~\eqref{eq:local_feature}, the aggregated feature map can be rewritten as
\begin{equation*}
    \bar{\mathbf{f}}=\frac{1}{K}\sum_k\left(\mathbf{P}_{k}\mathbf{g} + \mathbf{w}_k\right)=\bar{\mathbf{P}}\mathbf{g} + \frac{1}{K}\sum_k\mathbf{w}_k,
\end{equation*}
where $\bar{\mathbf{P}} = \frac{1}{K}\sum_k\mathbf{P}_{k}$. According to~\eqref{eq:uniform_prior}, the first term $\bar{\mathbf{P}}\mathbf{g}$ has a distribution of $\mathrm{Pr}(\bar{\mathbf{P}}\mathbf{g}=\bar{\mathbf{P}}\bmu_{\ell}) = \frac{1}{L}, \ \forall \ell$. At the same time, the summation of the i.i.d. Gaussian sensing noise, $\frac{1}{K}\sum_k\mathbf{w}_k$, follows a Gaussian distribution, say $\frac{1}{K}\sum_k\mathbf{w}_k\sim\mathcal{N}(\mathbf{0},\frac{1}{K}\mathbf{C})$. Hence, the overall distribution of $\bar{\mathbf{f}}$ is a Gaussian mixture with the same priors and the $\ell$-th Gaussian component has the mean $\bar{\mathbf{P}}\bmu_{\ell}$ and covariance $\mathbf{C}/K$, which completes the proof.

\subsection{Proof of Proposition~\ref{Prop:uncertainty}}\label{Apdx:uncertainty}
Based on the definition given in~\eqref{eq:uncertainty}, the sensing uncertainty using the aggregated feature $\bar{\mathbf{f}}$ is given as 
\begin{equation*}
    H  = \mathsf{E}\left[-\sum_{\ell=1}^{L}\int\mathrm{Pr}\left(\bmu_{\ell}|\bar{\mathbf{f}}\right)\log\mathrm{Pr}\left(\bmu_{\ell}|\bar{\mathbf{f}}\right)p(\bar{\mathbf{f}})\mathrm{d}\bar{\mathbf{f}}\right].
\end{equation*}
Given the output uniqueness of the linear classifier in~\eqref{eq:ML_classifier}, there is a one-to-one mapping between $\bmu_{\ell}$ and $\bar{\mathbf{P}}\bmu_{\ell}$, leading to $\mathrm{Pr}\left(\bmu_{\ell}|\bar{\mathbf{f}}\right)=\mathrm{Pr}\left(\bar{\mathbf{P}}\bmu_{\ell}|\bar{\mathbf{f}}\right)$. Then, using the Bayes' theorem, the probability $\mathrm{Pr}\left(\bar{\mathbf{P}}\bmu_{\ell}|\bar{\mathbf{f}}\right)$ can be expressed as $\mathrm{Pr}\left(\bar{\mathbf{P}}\bmu_{\ell}|\bar{\mathbf{f}}\right) =p\left(\bar{\mathbf{f}}|\bar{\mathbf{P}}\bmu_{\ell}\right)p(\bar{\mathbf{P}}\bmu_{\ell})/p(\bar{\mathbf{f}})$, which is used to rewrite $ H$ as
\begin{equation*}
    \begin{aligned}
        H=&-\sum_{\ell=1}^{L}\int p\left(\bar{\mathbf{f}}|\bar{\mathbf{P}}\bmu_{\ell}\right)p(\bar{\mathbf{P}}\bmu_{\ell})\log\frac{p\left(\bar{\mathbf{f}}|\bar{\mathbf{P}}\bmu_{\ell}\right)p(\bar{\mathbf{P}}\bmu_{\ell})}{p(\bar{\mathbf{f}})}\mathrm{d}\bar{\mathbf{f}} \\
        \overset{(a)}{=}&\frac{1}{L}\sum_{\ell=1}^{L}\int p\left(\bar{\mathbf{f}}|\bar{\mathbf{P}}\bmu_{\ell}\right)\log\frac{\sum_{\ell^{\prime}}p\left(\bar{\mathbf{f}}|\bar{\mathbf{P}}\bmu_{\ell^{\prime}}\right)}{p\left(\bar{\mathbf{f}}|\bar{\mathbf{P}}\bmu_{\ell}\right)}\mathrm{d}\bar{\mathbf{f}}\\
        \overset{(b)}{=}&\frac{1}{L}\sum_{\ell=1}^{L}\int p\left(\mathbf{x}\right)\log\frac{\sum_{\ell^{\prime}}e^{\frac{-K}{2}(\mathbf{x}+\bar{\mathbf{P}}\bphi_{\ell,\ell^{\prime}})^{\top}\mathbf{C}^{-1}(\mathbf{x}+\bar{\mathbf{P}}\bphi_{\ell,\ell^{\prime}})}}{e^{\frac{-K}{2}\mathbf{x}^{\top}\mathbf{C}^{-1}\mathbf{x}}} \mathrm{d}\mathbf{x},
    \end{aligned}
\end{equation*}
where the PDF $p\left(\bar{\mathbf{f}}|\bar{\mathbf{P}}\bmu_{\ell}\right)$ is used in step (a), the integration variable is changed as $\mathbf{x} = \bar{\mathbf{f}}-\bmu_{\ell}$ with the resulting $p\left(\mathbf{x}\right) = \mathcal{N}\left(\mathbf{0},\frac{\mathbf{C}}{K}\right)$, $\bphi_{\ell,\ell^{\prime}} = \bmu_{\ell}-\bmu_{\ell^{\prime}}$ in step (b).
Then, the lower and upper bounds of $H$ are given respectively as follows.
\subsubsection{Lower bound} 
First, using the definition of $G_{\ell,\ell^{\prime}}$ in Lemma~\ref{Lemma:global_gain}, the sensing uncertainty is given as $    H= \frac{1}{L}\sum_{\ell=1}^{L}\int p\left(\mathbf{x}\right)\log\left[1 + \sum_{\ell^{\prime}\neq\ell}e^{-K\mathbf{x}^{\top}\mathbf{C}^{-1}\bar{\mathbf{P}}\bphi_{\ell,\ell^{\prime}}}e^{-\frac{1}{2}G_{\ell,\ell^{\prime}}}\right] \mathrm{d}\mathbf{x}$. Then, based on the convexity of log-sum-exp functions, using Jensen's inequality gives a lower bound of $H$ as
\begin{equation*}
    \begin{aligned}
        H
        &\geq \frac{1}{L}\sum_{\ell=1}^{L}\log\left[1 + \sum_{\ell^{\prime}\neq\ell}e^{-\int p\left(\mathbf{x}\right)K\mathbf{x}^{\top}\mathbf{C}^{-1}\bar{\mathbf{P}}\bphi_{\ell,\ell^{\prime}}\mathrm{d}\mathbf{x}}e^{-\frac{1}{2}G_{\ell,\ell^{\prime}}}\right],
    \end{aligned}
\end{equation*}
where the integral in the exponential term is computed as zero and the final lower bound is obtained.
\subsubsection{Upper bound}
To derive the upper bound of the sensing uncertainty, we first define rewrite $H$ derived before as 
\begin{equation*}
    \begin{aligned}
        H=&\frac{1}{L}\sum_{\ell=1}^{L}\int p\left(\mathbf{x}\right)\log\frac{\sum_{\ell^{\prime}}e^{\frac{-K}{2}(\mathbf{x}+\bar{\mathbf{P}}\bphi_{\ell,\ell^{\prime}})^{\top}\mathbf{C}^{-1}(\mathbf{x}+\bar{\mathbf{P}}\bphi_{\ell,\ell^{\prime}})}}{e^{\frac{-K}{2}\mathbf{x}^{\top}\mathbf{C}^{-1}\mathbf{x}}} \mathrm{d}\mathbf{x}\\
        \leq& \frac{1}{L}\sum_{\ell=1}^{L}\log \left[1+\sum_{\ell^{\prime}\neq \ell}\exp\left(-\frac{1}{a+2}G_{\ell,\ell^{\prime}}\right)\right] \\
        &+\frac{M}{a}\log e-\frac{M}{2}\log (1+2/a),
    \end{aligned}
\end{equation*}
where $a$ is a positive constant and the inequality is obtained by using the log-concavity and Jensen's inequality. Then, let $a=cM$ with $c> 0$, there is $\frac{M}{a}\log e-\frac{M}{2}\log (1+2/a) \leq \log \frac{e^{\frac{1}{c}}}{1+\frac{2}{cM}\frac{M}{2}}=\log \frac{e^{\frac{1}{c}}}{1+\frac{1}{c}}$, where the second step is based on the well-known Bernoulli's inequality. This gives the upper bound of $H$.

\subsection{Proof of Proposition~\ref{Prop:surrogate}}\label{Apdx:relaxed_up}
Let $\{D_{\ell,\ell^{\prime}}\}$ be aggregated into a $L (L-1)$-dimensional vector $\mathbf{u}$ and define the function $U_{\kappa}(K,\mathbf{u}) = \frac{1}{L}\sum_{\ell=1}^{L}\log\left[1+\sum_{\ell^{\prime}\neq\ell}\exp\left(-\kappa D_{\ell,\ell^{\prime}}K\right)\right]$ for simplicity. Using the first-order approximation, $U_{\kappa}(K,\mathbf{u})$ can be rewritten as $U_{\kappa}(K,\mathbf{u})= U_{\kappa}(K,\mathbf{c}) + (\mathbf{u}-\mathbf{c})^{\top}\nabla_{\mathbf{u}}U_{\kappa}(K,\mathbf{u})|_{\mathbf{u}=\mathbf{c}} + \mathcal{O}\left((\mathbf{u}-\mathbf{c})^{\top}\mathcal{H}_{\mathbf{u}}U_{\kappa}(K,\mathbf{u})|_{\mathbf{u}=\mathbf{c}}(\mathbf{u}-\mathbf{c})\right)$, where $\mathbf{c}$ denotes an arbitrary constant vector, $\nabla_{\mathbf{u}}U_{\kappa}(K,\mathbf{u})$ and $\mathcal{H}_{\mathbf{u}}U_{\kappa}(K,\mathbf{u})$ denote the  derivative and Hessian of $U_{\kappa}(K,\mathbf{u})$ w.r.t. $\mathbf{u}$, respectively. Then, the gradient values of $U_{\kappa}(K,\mathbf{u})$ are expressed as
\begin{equation*}
        \frac{\partial U_{\kappa}(K,\mathbf{u})}{\partial  D_{\ell,\ell^{\prime}}} = -\frac{\kappa K}{L}\frac{\exp\left(-\kappa D_{\ell,\ell^{\prime}}K\right)}{1+\sum_{l\neq \ell}\exp\left(-\kappa D_{\ell,l} K\right)}.
\end{equation*}
Hence, let the constant vector $\mathbf{c}$ be $\mathbf{c} =  \bar{D}\mathbf{1}$ with $ \bar{D}$ denoting the average version of $\{D_{\ell,\ell^{\prime}}\}$ as stated before. There is then $\frac{\partial U_{\kappa}(K,\mathbf{u})}{\partial  D_{\ell,\ell^{\prime}}}|_{ D_{\ell,\ell^{\prime}} =  \bar{D}}=-\frac{\kappa K}{L}\frac{\exp\left(-\kappa \bar{D}K\right)}{1+\sum_{l\neq \ell}\exp\left(-\kappa \bar{D}K\right)}$, based on which the residual term can be further written as $\mathcal{O}\left((\mathbf{u}-\mathbf{c})^{\top}\mathcal{H}_{\mathbf{u}}U_{\kappa}(K,\mathbf{u})|_{\mathbf{u}=\mathbf{c}}(\mathbf{u}-\mathbf{c})\right) = \mathcal{O}(\frac{1}{L(L-1)}\sum_{\ell}\sum_{\ell\neq\ell^{\prime}}(D_{\ell,\ell^{\prime}}-\bar{D})^2)$.
It then follows that $U_{\kappa}(K,\mathbf{u})$ can be first-order approximated as
\begin{equation*}
    \begin{aligned}
        U_{\kappa}(K,\mathbf{u})=& U_{\kappa}(K,\bar{D}\mathbf{1}) + \underbrace{(\mathbf{u}- \bar{D}\mathbf{1})^{\top}\nabla_{\mathbf{u}}U_{\kappa}(K,\mathbf{u})|_{\mathbf{u}=\bar{D}\mathbf{1}}}_{=0} \\
        &+ \mathcal{O}(\frac{1}{L(L-1)}\sum_{\ell}\sum_{\ell\neq\ell^{\prime}}(D_{\ell,\ell^{\prime}}-\bar{D})^2).
    \end{aligned}
\end{equation*}
Based on the definition, $U_{\kappa}(K,\bar{D}\mathbf{1})$ can be expressed as $U_{\kappa}(K,\bar{D}\mathbf{1}) = \frac{1}{L}\sum_{\ell=1}^{L}\log\left[1+\sum_{\ell^{\prime}\neq \ell}\exp\left(-\kappa \bar{D}K\right)\right]$ that gives the final result.

\subsection{Proof of Lemma~\ref{Lemma:aircomp_feature}}\label{Apdx:aircomp_feature}
Using~\eqref{eq:local_feature} and~\eqref{eq:post_process}, the efficient aggregated feature map can be rewritten as
\begin{equation*}
    \begin{aligned}
        \tilde{\mathbf{f}}&=\bar{\mathbf{P}}\mathbf{g} + \frac{1}{K}\sum_k\mathbf{w}_k+ \tilde{\mathbf{z}},
    \end{aligned}
\end{equation*}
where the noise term $\tilde{\mathbf{z}}=\frac{1}{K}\sum_k\mathbf{w}_k+ \frac{1}{K}\Re\{\mathbf{Z}^H\mathbf{b}\}$ is a summation of real independent Gaussian vectors and thus has a Gaussian distribution. Since both $\mathbf{w}_k$ and $\mathbf{Z}^{\top}$ are zero-mean, there is $\mathsf{E}\left[\tilde{\mathbf{z}}\right] = \mathbf{0}$.
Furthermore, $\mathbf{w}_k$ and $\mathbf{Z}^{\top}$ are mutually independent, leading the covariance matrix of $\tilde{\mathbf{z}}$ computed as
\begin{equation*}
    \begin{aligned}
        \mathsf{E}\left[\tilde{\mathbf{z}}\tilde{\mathbf{z}}^{\top}\right] &= \frac{1}{K^2}\mathsf{E}\left[\sum_{k,j}\mathbf{w}_k\mathbf{w}_j^{\top}\right] + \frac{1}{K^2}\mathsf{E}\left[\Re\{\mathbf{Z}^H\mathbf{b}\}\Re\{\mathbf{b}^{\top}\mathbf{Z}^*\}\right]\\
        &= \frac{1}{K}\mathbf{C} + \frac{\sigma^2}{2K^2}\Vert\mathbf{b}\Vert^2.
    \end{aligned}
\end{equation*}
Still, there is $p(\bar{\mathbf{P}}\mathbf{g}=\bar{\mathbf{P}}\bmu_{\ell}) = \frac{1}{L}, \ \forall \ell$, meaning that $\tilde{\mathbf{f}}^{\mathsf{re}}$ follows a GMM with the uniform priors and the $\ell$-th Gaussian component having the mean $\bar{\mathbf{P}}\bmu_{\ell}$ and covariance $\frac{1}{K}\mathbf{C} + \frac{\sigma^2\Vert\mathbf{b}\Vert^2}{2K^2}\mathbf{I}_M$.

\subsection{Proof of Lemma~\ref{Lemma:monotonical}}\label{Apdx:monotonical}
Straightforwardly, $H_{\mathsf{s}}(\gamma_{\mathsf{air}})$ is monotonically decreasing w.r.t. $\tilde{D}(\gamma_{\mathsf{air}})$. Hence, we prove the Lemma~\ref{Lemma:monotonical} by showing that the derivative of $\tilde{D}(\gamma_{\mathsf{air}})$ w.r.t. $\gamma_{\mathsf{air}}$ is always positive given any $\mathbf{b}$. To this end, the chain rule of matrix derivative is used as
\begin{equation*}
    \frac{\partial \tilde{D}(\gamma_{\mathsf{air}})}{\partial \gamma_{\mathsf{air}}} = \mathsf{Tr}\left(\frac{\partial \tilde{D}(\gamma_{\mathsf{air}})}{\partial \mathbf{A}}\cdot\left(\frac{\partial \mathbf{A}}{\partial \gamma_{\mathsf{air}}}\right)^{\top}\right),
\end{equation*}
where the matrix $\mathbf{A} = \left(\mathbf{C} + \frac{K}{\gamma_{\mathsf{air}}}\mathbf{I}_M\right)^{-1}$ for ease of notation. Therein, using the trace property, $\frac{\partial \tilde{D}(\gamma_{\mathsf{air}})}{\partial \mathbf{A}}$ is computed as $\frac{\partial \tilde{D}(\gamma_{\mathsf{air}})}{\partial \mathbf{A}} = \bar{\mathbf{P}}\mathbf{D}\bar{\mathbf{P}}$. The other derivative term, $\frac{\partial \mathbf{A}}{\partial \gamma_{\mathsf{air}}}$, is computed by using the matrix inverse' derivative as $        \frac{\partial \mathbf{A}}{\partial \gamma_{\mathsf{air}}} = -\mathbf{A}\frac{\partial \mathbf{A}^{-1}}{\partial \gamma_{\mathsf{air}}}\mathbf{A} = \frac{K}{\gamma_{\mathsf{air}}^2}\mathbf{A}^2$. As a result, we have the derivative $ \frac{\partial \tilde{D}(\gamma_{\mathsf{air}})}{\partial \gamma_{\mathsf{air}}}= \frac{K}{\gamma_{\mathsf{air}}^2}\mathsf{Tr}\left(\bar{\mathbf{P}}\mathbf{D}\bar{\mathbf{P}}\mathbf{A}^2\right) = \frac{K}{\gamma_{\mathsf{air}}^2}\mathsf{Tr}\left(\mathbf{A}\bar{\mathbf{P}}\mathbf{D}\bar{\mathbf{P}}\mathbf{A}\right)$, where the argument of the trace function is a semi-positive Hermitian matrix, leading to $ \frac{\partial \tilde{D}(\gamma_{\mathsf{air}})}{\partial \gamma_{\mathsf{air}}} \geq 0$. This completes the proof.

\subsection{Proof of Lemma~\ref{Lemma:asymp_dist}}\label{Apdx:asymp_dist}
To obtain the distribution of $\zeta_{\mathsf{air}}$, we first rewrite the expression of the data vector $\mathbf{s}$ in~\eqref{eq:post_process} to rewrite $ \Vert\mathbf{b}^{\star}\Vert$. Specifically, since $\mathbf{b}^{\star} = \Vert\mathbf{b}^{\star}\Vert\mathbf{v}$ with $\mathbf{v}$ being the first eigenvector of the channel matrix $\mathbf{H}$, there is
\begin{equation*}
\begin{aligned}
    \mathbf{s} &= \mathbf{Y}^H\mathbf{b}^{\star} \\
    &=\mathbf{Z}^H\mathbf{b}^{\star} + \mathbf{X}\mathsf{diag}(\rho_1^*,\cdots,\rho_K^*)\mathbf{q}_1\sqrt{\lambda_1}\Vert\mathbf{b}^{\star}\Vert,
\end{aligned}
\end{equation*}
where $\mathbf{X}=[\mathbf{x}_1^H,\cdots,\mathbf{x}_K^H]$, $\mathbf{q}_1$ and $\lambda_1$ denotes the first eigenvector and the first eigenvalue of $\mathbf{H}^H\mathbf{H}$. Therefore, to realize AirComp, the transmit power control of sensor $k$ can be rewritten as $\rho_k = \frac{1}{q_{k,1}^*\sqrt{\lambda_1}\Vert\mathbf{b}^{\star}\Vert}$,
where $q_{k,1}$ denotes the $k$-th element of $\mathbf{q}_1$. To minimize $\Vert\mathbf{b}^{\star}\Vert$ under the power constraint $\rho_k^2\leq \frac{P}{\nu^2}$, $\forall k$, there is $\Vert\mathbf{b}^{\star}\Vert = \max_k\ \frac{\nu^2}{P}\frac{1}{q_{k,1}^2\lambda_1} = \frac{\nu^2}{P}\frac{1}{\lambda_1\min_k\ q_{k,1}^2}$. Comparing the above result, the scaled effective channel gain can be re-expressed as $\zeta_{\mathsf{air}} =K\lambda_1\min_k\ q_{k,1}^2$.
Hence, $\zeta_{\mathsf{air}}$ can be characterized by investigating the asymptotic values of $\lambda_1$ and $\min_k\ q_{k,1}^2$. Specifically, using the asymptotic spectrum of random matrices~\cite{Edelman1989}, given $N=\omega K$, there is
\begin{equation*}
    \frac{1}{K}\lambda_1\rightarrow (1+\sqrt{\omega})^2,\ K\rightarrow \infty.
\end{equation*}
At the same time, since $\mathbf{H}$ is composed of i.i.d. Gaussian elements, the eigenvector $\mathbf{q}_1$ is isotropically distributed on a $K$-dimensional complex hypersphere and independent of the eigenvalue $\lambda_1$. Furthermore, it has the same distribution as a normalized complex Gaussian vector $\mathbf{h} = [h_1,\cdots,h_K]$ with $h_k\sim\mathcal{CN}(0,1)$, say $\mathbf{q}_1\overset{d}{=} \frac{\mathbf{h}}{\Vert\mathbf{h}\Vert}$. Therefore, we have 
\begin{equation*}
    K\cdot\min_k\ q_{k,1}^2\overset{d}{=}\frac{K}{\Vert\mathbf{h}\Vert^2}\min_k\ h_k^2,
\end{equation*}
where $\frac{K}{\Vert\mathbf{h}\Vert^2} = 1$ as $K\rightarrow \infty$, which means that $K\cdot\min_k\ q_{k,1}^2$ asymptotically equals to $\min_k\ h_k^2$. Clearly, $2h_k^2$ is a chi-square random variable with the degrees of freedom $2$ and is also exponentially distributed with the parameter of $\frac{1}{2}$, say $h_k^2\sim \mathrm{Exp}(\frac{1}{2})$. Using the order statistics of exponential distributions, we have $\min_k\ 2h_k^2\sim\mathrm{Exp}(\frac{K}{2})$. Concluding the above results gives the asymptotic distribution of $\zeta_{\mathsf{air}} \rightarrow \frac{1}{2}K(1+\sqrt{\omega})^2\zeta^{\prime}$ with $\zeta^{\prime}\sim\mathrm{Exp}\left(\frac{K}{2}\right)$. Finally, using the scaling property of exponential distributions obtains the distribution of $\zeta_{\mathsf{air}}$. 

\subsection{Proof of Lemma~\ref{Lemma:independence}}\label{Apdx:independence}
To obtain the distribution of the norm of the $k$-th ZF receiver, let $\mathbf{U}_k$ be the $(K-1)$-dimensional principal eigenspace of the matrix $[\mathbf{h}_1,\cdots,\mathbf{h}_{k-1},\mathbf{h}_{k+1},\cdots,\mathbf{h}_{K}]$ and $\Vert\mathbf{b}_k\Vert^2$ can be rewritten as $ \Vert\mathbf{b}_k\Vert^2 = \mathbf{e}_k^H(\mathbf{H}^H\mathbf{H})^{-1}\mathbf{e}_k=\frac{1}{\mathbf{h}_k\mathbf{U}_k^{\perp}(\mathbf{U}_k^{\perp})^H\mathbf{h}_k}$,
where $\mathbf{U}_k^{\perp}$ denotes the orthogonal compliment of $\mathbf{U}_k$ and $(\mathbf{U}_k^{\perp})^H\mathbf{U}_k^{\perp} = \mathbf{I}_{N-K+1}$. Since $\mathbf{h}_k$ is composed of i.i.d. $\mathcal{CN}(0,1)$ elements, it is isotropically distributed on the $N$-dimensional complex unit hypersphere, i.e., $p(\mathbf{h}_k) = p(\mathbf{Q}\mathbf{h}_k)$ for any $\mathbf{Q}$ satisfying $\mathbf{Q}^H\mathbf{Q} = \mathbf{Q}\mathbf{Q}^H=\mathbf{I}_N$. Hence, let $\mathbf{Q} = [\mathbf{U}_k, \mathbf{U}_k^{\perp}]$ and
\begin{equation*}
    \begin{aligned}
        \mathbf{h}_k\mathbf{U}_k^{\perp}(\mathbf{U}_k^{\perp})^H\mathbf{h}_k \overset{d}{=} \mathbf{h}_k^H\mathbf{Q}^H\mathbf{U}_k^{\perp}(\mathbf{U}_k^{\perp})^H\mathbf{Q}\mathbf{h}_k=\sum_{k=K}^Nh_{k,n}^2.
    \end{aligned}
\end{equation*}
Since $h_{k,n}\sim\mathcal{CN}(0,1)$, the summation $2\sum_{k=K}^Nh_{k,n}^2$ follows a chi-square distribution with $2(N-K+1)$ degrees of freedom, which completes the proof.
\bibliographystyle{IEEEtran}

\bibliography{Ref}

\begin{thebibliography}{10}
\providecommand{\url}[1]{#1}
\csname url@samestyle\endcsname
\providecommand{\newblock}{\relax}
\providecommand{\bibinfo}[2]{#2}
\providecommand{\BIBentrySTDinterwordspacing}{\spaceskip=0pt\relax}
\providecommand{\BIBentryALTinterwordstretchfactor}{4}
\providecommand{\BIBentryALTinterwordspacing}{\spaceskip=\fontdimen2\font plus
\BIBentryALTinterwordstretchfactor\fontdimen3\font minus
  \fontdimen4\font\relax}
\providecommand{\BIBforeignlanguage}[2]{{%
\expandafter\ifx\csname l@#1\endcsname\relax
\typeout{** WARNING: IEEEtran.bst: No hyphenation pattern has been}%
\typeout{** loaded for the language `#1'. Using the pattern for}%
\typeout{** the default language instead.}%
\else
\language=\csname l@#1\endcsname
\fi
#2}}
\providecommand{\BIBdecl}{\relax}
\BIBdecl

\bibitem{edgeAI_Letaief}
K.~B. {Letaief}, Y.~{Shi}, J.~{Lu}, and J.~{Lu}, ``Edge artificial intelligence
  for {6G}: Vision, enabling technologies, and applications,'' \emph{IEEE J.
  Sel. Areas Commun.}, vol.~40, no.~1, pp. 5--36, Jan. 2022.

\bibitem{WirelessmeetsML}
G.~{Zhu}, D.~{Liu}, Y.~{Du}, C.~{You}, J.~{Zhang}, and K.~{Huang}, ``Toward an
  intelligent edge: Wireless communication meets machine learning,'' \emph{IEEE
  Commun. Mag.}, vol.~58, no.~1, p. 19–25, 2020.

\bibitem{Sensing_in_6G}
Z.~{Feng}, Z.~{Wei}, X.~{Chen}, H.~{Yang}, Q.~{Zhang}, and P.~Zhang, ``Joint
  communication, sensing, and computation enabled {6G} intelligent machine
  system,'' \emph{IEEE Netw.}, vol.~35, no.~6, pp. 34--42, 2021.

\bibitem{Letaief_CollaSensing2023}
L.~{Xie}, S.~{Song}, Y.~C. {Eldar}, and K.~B. {Letaief}, ``Collaborative
  sensing in perceptive mobile networks: Opportunities and challenges,''
  \emph{IEEE Wireless Commun.}, vol.~30, no.~1, pp. 16--23, 2023.

\bibitem{distributedAI}
S.~{Duan}, D.~{Wang}, J.~{Ren}, F.~{Lyu}, Y.~{Zhang}, H.~{Wu}, and X.~{Shen},
  ``Distributed artificial intelligence empowered by end-edge-cloud computing:
  A survey,'' \emph{IEEE Commun. Surv. Tutor.}, vol.~25, no.~1, pp. 591--624,
  Firstquarter 2023.

\bibitem{airpooling2023}
Z.~{Liu}, Q.~{Lan}, A.~E. {Kalør}, P.~{Popovski}, and K.~{Huang},
  ``Over-the-air multi-view pooling for distributed sensing,'' [Online]
  \url{https://arxiv.org/pdf/2302.09771.pdf}, 2023.

\bibitem{MVCNN}
H.~{Su}, S.~{Maji}, E.~{Kalogerakis}, and E.~{Learned-Miller}, ``Multi-view
  convolutional neural networks for {3D} shape recognition,'' in \emph{Proc.
  IEEE Int. Conf. Comput. Vision (ICCV)}, Santiago, Chile, Dec. 7–13 2015.

\bibitem{CV_for_locolization2022}
L.~{Bai}, Y.~{Yang}, M.~{Chen}, C.~{Feng}, C.~{Guo}, W.~{Saad}, and S.~{Cui},
  ``Computer vision-based localization with visible light communications,''
  \emph{IEEE Trans. Wireless Commun.}, vol.~21, no.~3, pp. 2051--2065, Mar.
  2022.

\bibitem{Multiagent_perception}
Y.~{Li}, S.~{Ren}, P.~{Wu}, S.~{Chen}, C.~{Feng}, and W.~{Zhang}, ``Learning
  distilled collaboration graph for multi-agent perception,'' in \emph{Proc.
  Adv. in Neural Inf. Process. Syst. (NeurIPS)}, Virtual-only, Dec. 6-12 2021.

\bibitem{DIB_JZhang}
J.~{Shao}, Y.~{Mao}, and J.~{Zhang}, ``Task-oriented communication for
  multi-device cooperative edge inference,'' \emph{IEEE Trans. Wireless
  Commun.}, vol.~22, no.~1, p. 73–87, Jan. 2023.

\bibitem{Accelarate_edgeAI}
E.~{Li}, L.~{Zeng}, Z.~{Zhou}, and X.~{Chen}, ``Edge {AI}: On-demand
  accelerating deep neural network inference via edge computing,'' \emph{IEEE
  Trans. Wireless Commun.}, vol.~19, no.~1, p. 447–457, 2020.

\bibitem{SZhou_split_energy}
W.~Shi, S.~Zhou, Z.~Niu, M.~Jiang, and L.~Geng, ``Multiuser co-inference with
  batch processing capable edge server,'' \emph{IEEE Trans. Wireless Commun.},
  vol.~22, no.~1, pp. 286--300, 2023.

\bibitem{Bennis_split_latency}
C.-F. Liu, M.~Bennis, M.~Debbah, and H.~V. Poor, ``Dynamic task offloading and
  resource allocation for ultra-reliable low-latency edge computing,''
  \emph{IEEE Trans. Commun.}, vol.~67, no.~6, pp. 4132--4150, 2019.

\bibitem{MVsensing_graph}
X.~Wei, R.~Yu, and J.~Sun, ``Learning view-based graph convolutional network
  for multi-view {3D} shape analysis,'' \emph{IEEE Pattern Anal. Mach.
  Intell.}, vol.~45, no.~6, pp. 7525--7541, 2023.

\bibitem{Larsson_grantfree}
K.~Senel and E.~G. Larsson, ``Grant-free massive {MTC}-enabled massive {MIMO}:
  A compressive sensing approach,'' \emph{IEEE Trans. Commun.}, vol.~66,
  no.~12, pp. 6164--6175, 2018.

\bibitem{WeiYu_grantfree}
L.~Liu, E.~G. Larsson, W.~Yu, P.~Popovski, C.~Stefanovic, and E.~de~Carvalho,
  ``Sparse signal processing for grant-free massive connectivity: A future
  paradigm for random access protocols in the internet of things,'' \emph{IEEE
  Signal Process. Mag.}, vol.~35, no.~5, pp. 88--99, 2018.

\bibitem{Chen2016Multiview}
X.~Chen, H.~Ma, J.~Wan, B.~Li, and T.~Xia, ``Multi-view {3D} object detection
  network for autonomous driving,'' in \emph{Proc. IEEE Conf. Comput. Vis.
  Pattern Recognit. (CVPR)}, Las Vegas, NV, USA, Jun. 2016.

\bibitem{3GPP_AItraffic}
3GPP, ``Release 18: Study on traffic characteristics and performance
  requirements for {AI/ML} model transfer,'' \emph{3GPP TR 22.874, version
  18.2.0}, May 2021.

\bibitem{When2com}
Y.-C. {Liu}, J.~{Tian}, N.~{Glaser}, and Z.~{Kira}, ``When2com: Multi-agent
  perception via communication graph grouping,'' in \emph{Proc. IEEE/CVF Conf.
  Comput. Vision Pattern Recogn. (CVPR)}, Seattle, WA, USA, Jun. 26 – 29
  2020.

\bibitem{feature_selection}
M.~{Singhal}, V.~{Raghunathan}, and A.~{Raghunathan}, ``Communication-efficient
  view-pooling for distributed multi-view neural networks,'' in \emph{Proc.
  Des. Autom. Test Eur. Conf. \& Exh. (DATE)}, Grenoble, France, Mar. 9–13
  2020.

\bibitem{GXZhuAirComp2019}
G.~Zhu and K.~Huang, ``{MIMO} over-the-air computation for high-mobility
  multimodal sensing,'' \emph{IEEE Internet Things J.}, vol.~6, no.~4, pp.
  6089--6103, 2019.

\bibitem{Air_distlearning}
M.~Mohammadi~Amiri and D.~Gündüz, ``Machine learning at the wireless edge:
  Distributed stochastic gradient descent over-the-air,'' \emph{IEEE Trans.
  Signal Process.}, vol.~68, pp. 2155--2169, 2020.

\bibitem{Air_edgeinference}
S.~F. {Yilmaz}, B.~{Hasırcıoglu}, and D.~{Gündüz}, ``Over-the-air ensemble
  inference with model privacy,'' in \emph{Proc. IEEE Int. Symp. Inf. Theory
  (ISIT)}, Espoo, Finland, Jun. 26 – Jul. 1 2022.

\bibitem{Task_oriented_DZ}
D.~Wen, X.~Jiao, P.~Liu, G.~Zhu, Y.~Shi, and K.~Huang, ``Task-oriented
  over-the-air computation for multi-device edge {AI},'' \emph{IEEE Trans.
  Wireless Commun.}, pp. 1--1, 2023.

\bibitem{chen2022analog}
X.~Chen, E.~G. Larsson, and K.~Huang, ``Analog {MIMO} communication for
  one-shot distributed principal component analysis,'' \emph{IEEE Trans. Signal
  Process.}, vol.~70, pp. 3328--3342, 2022.

\bibitem{Marzetta_analogfb}
T.~Marzetta and B.~Hochwald, ``Fast transfer of channel state information in
  wireless systems,'' \emph{IEEE Trans. Signal Process.}, vol.~54, no.~4, pp.
  1268--1278, 2006.

\bibitem{MZChen_JSACoverview}
M.~Chen, D.~Gündüz, K.~Huang, W.~Saad, M.~Bennis, A.~V. Feljan, and H.~V.
  Poor, ``Distributed learning in wireless networks: Recent progress and future
  challenges,'' \emph{IEEE J. Sel. Areas Commun.}, vol.~39, no.~12, pp.
  3579--3605, 2021.

\bibitem{Ding2020TWC}
K.~Yang, T.~Jiang, Y.~Shi, and Z.~Ding, ``Federated learning via over-the-air
  computation,'' \emph{IEEE Trans. Wireless Commun.}, vol.~19, no.~3, pp.
  2022--2035, 2020.

\bibitem{Eldar2021TSP}
T.~Sery, N.~Shlezinger, K.~Cohen, and Y.~C. Eldar, ``Over-the-air federated
  learning from heterogeneous data,'' \emph{IEEE Trans. Signal Process.},
  vol.~69, pp. 3796--3811, 2021.

\bibitem{MXTao_powercontrol}
N.~Zhang and M.~Tao, ``Gradient statistics aware power control for over-the-air
  federated learning,'' \emph{IEEE Trans. Wireless Commun.}, vol.~20, no.~8,
  pp. 5115--5128, 2021.

\bibitem{GX_broadband}
G.~Zhu, Y.~Wang, and K.~Huang, ``Broadband analog aggregation for low-latency
  federated edge learning,'' \emph{IEEE Trans. Wireless Commun.}, vol.~19,
  no.~1, pp. 491--506, 2020.

\bibitem{ISACC_air2023}
X.~{Li}, F.~{Liu}, Z.~{Zhou}, G.~{Zhu}, S.~{Wang}, K.~{Huang}, and Y.~{Gong},
  ``Integrated sensing, communication, and computation over-the-air: {MIMO}
  beamforming design,'' \emph{IEEE Trans. Wireless Commun.}, vol.~22, no.~8,
  pp. 5383--5398, 2023.

\bibitem{ISCC_DZ}
D.~Wen, P.~Liu, G.~Zhu, Y.~Shi, J.~Xu, Y.~C. Eldar, and S.~Cui, ``Task-oriented
  sensing, computation, and communication integration for multi-device edge
  {AI},'' \emph{IEEE Trans. Wireless Commun.}, pp. 1--1, 2023.

\bibitem{Air_diversity}
S.~Tang, P.~Popovski, C.~Zhang, and S.~Obana, ``Multi-slot over-the-air
  computation in fading channels,'' \emph{IEEE Trans. Wireless Commun.},
  vol.~22, no.~10, pp. 6766--6777, 2023.

\bibitem{Multiview_analysis_2018}
Z.~Ding and Y.~Fu, ``Robust multiview data analysis through collective low-rank
  subspace,'' \emph{IEEE Trans. Neural Netw. Learn. Syst.}, vol.~29, no.~5, pp.
  1986--1997, 2018.

\bibitem{Multiview_analysis_2019}
P.~Hu, D.~Peng, Y.~Sang, and Y.~Xiang, ``Multi-view linear discriminant
  analysis network,'' \emph{IEEE Trans. Image Process.}, vol.~28, no.~11, pp.
  5352--5365, 2019.

\bibitem{statistic_learning_book}
T.~{Hastie}, R.~{Tibshirani}, and J.~{Friedman}, \emph{The Elements of
  Statistical Learning: Data mining, Inference, and Prediction}.\hskip 1em plus
  0.5em minus 0.4em\relax New York, NY, USA: Springer Science \& Business
  Media, 2009.

\bibitem{GMM_pose_esti}
M.~{Ye} and R.~{Yang}, ``Real-time simultaneous pose and shape estimation for
  articulated objects using a single depth camera,'' in \emph{Proc. IEEE Conf.
  Comput. Vis. Pattern Recognit. (CVPR)}, Columbus, OH, USA, Jun. 23-28 2014.

\bibitem{ProgressFTX_Qiao}
Q.~Lan, Q.~Zeng, P.~Popovski, D.~Gündüz, and K.~Huang, ``Progressive feature
  transmission for split classification at the wireless edge,'' \emph{IEEE
  Trans. Wireless Commun.}, vol.~22, no.~6, pp. 3837--3852, 2023.

\bibitem{Entropy2ErrorPorb}
M.~Feder and N.~Merhav, ``Relations between entropy and error probability,''
  \emph{IEEE Trans. Inf. Theory}, vol.~40, no.~1, pp. 259--266, 1994.

\bibitem{Taesang2006}
T.~Yoo and A.~Goldsmith, ``On the optimality of multiantenna broadcast
  scheduling using zero-forcing beamforming,'' \emph{IEEE J. Sel. Areas
  Commun.}, vol.~24, no.~3, pp. 528--541, 2006.

\bibitem{Mahalanobis1936}
P.~C. Mahalanobis, ``On the generalized distance in statistics,'' \emph{Proc.
  National Inst. Sci. India}, vol.~2, no.~1, pp. 49--55, 1936.

\bibitem{MatrixAnalysis}
R.~A. {Horn} and C.~R. {Johnson}, \emph{Matrix Analysis}.\hskip 1em plus 0.5em
  minus 0.4em\relax Cambridge Univ. Press, 2012.

\bibitem{Edelman1989}
A.~Edelman, ``Eigenvalues and condition numbers of random matrices,'' Ph.D.
  dissertation, Department of Mathematics, Massachusetts Institute of
  Technology, Cambridge, MA, 1989.

\end{thebibliography}

\end{document}